\documentclass[11pt]{article}
\usepackage{fullpage}
\usepackage{amsmath,amsfonts,amsthm,amssymb,multirow}
\usepackage{times}
\usepackage{graphicx}
\usepackage{floatpag}
\usepackage{algorithm}
\usepackage[noend]{algpseudocode}
\usepackage{bbold}
\usepackage{enumitem}
\usepackage{subcaption} 
\usepackage{calc}
\usepackage{tikz}
\usetikzlibrary{decorations.markings}
\tikzstyle{vertex}=[circle, draw, inner sep=0pt, minimum size=6pt]

\usepackage{graphicx}
\usepackage{tabularx}
\makeatletter
\newcommand{\multiline}[1]{%
  \begin{tabularx}{\dimexpr\linewidth-\ALG@thistlm}[t]{@{}X@{}}
    #1
  \end{tabularx}
}
\makeatother

\makeatletter
\def\BState{\State\hskip-\ALG@thistlm}
\makeatother

\usepackage[compact]{titlesec}
\titlespacing{\section}{0pt}{3ex}{2ex}
\titlespacing{\subsection}{0pt}{2ex}{1ex}
\titlespacing{\subsubsection}{0pt}{0.5ex}{0ex}

\newtheorem{theorem}{Theorem}[section]

\newtheorem{proposition}{Proposition}

\newtheorem{lemma}{Lemma}[section]

\makeatletter
\let\c@fconjecture\c@conjecture
\makeatother

\makeatletter
\let\c@fconj\c@conj
\makeatother

\def \eps {\varepsilon}

\newcommand{\ignore}[1]{}

\def \poly { \text{\rm poly~} }

\def\tO{\tilde{O}}

\title{Tight Approximation Algorithms for Bichromatic Graph Diameter and Related Problems}
\author{Mina Dalirrooyfard \footnote{\texttt{minad@mit.edu}, Supported by an Akamai Presidential Fellowship and NSF Grant.}\\MIT%funding info for Mina?
\and Virginia Vassilevska Williams \footnote{\texttt{virgi@mit.edu}, Supported by an NSF CAREER Award, NSF
Grants CCF-1417238, CCF-1528078 and CCF-1514339, a BSF Grant
BSF:2012338 and a Sloan Research Fellowship.}\\MIT
\and Nikhil Vyas \footnote{\texttt{nikhilv@mit.edu}, Supported by an Akamai Presidential Fellowship and NSF Grant CCF-1552651.}\\MIT
\and Nicole Wein \footnote{\texttt{nwein@mit.edu}, Supported by an NSF Graduate Fellowship and NSF Grant CCF-1514339}\\ MIT}
\date{}

\begin{document}
\maketitle
\begin{abstract}
% Proximity problems such as Diameter (also called Farthest Pair) are fundamental both in graphs and in point sets (say in Euclidean space). A very well studied problem in geometry is the Bichromatic version: each point is colored one of two colors and the goal is to find (say) the Farthest Pair of points of different colors. The Bichromatic version of Diameter is just as interesting in graphs, though seemingly not as well studied, largely since there are strong conditional hardness results for the standard Diameter problem, and the Bichromatic version is only harder. 

Some of the most fundamental and well-studied graph parameters are the Diameter (the largest shortest paths distance) and Radius (the smallest distance for which a ``center'' node can reach all other nodes).
The natural and important $ST$-variant considers two subsets $S$ and $T$ of the vertex set and lets the $ST$-diameter be the maximum distance between a node in $S$ and a node in $T$, and the $ST$-radius be the minimum distance for a node of $S$ to reach all nodes of $T$. The {\em bichromatic} variant is the special case in which $S$ and $T$ partition the vertex set.

In this paper we present a comprehensive study of the \emph{approximability} of $ST$ and Bichromatic Diameter, Radius, and Eccentricities, and variants, in graphs with and without directions and weights. We give the first nontrivial approximation algorithms for most of these problems, including time/accuracy trade-off upper and lower bounds. We show that nearly \emph{all} of our obtained bounds are tight under the Strong Exponential Time Hypothesis (SETH), or the related Hitting Set Hypothesis. 

For instance, for Bichromatic Diameter in undirected weighted graphs with $m$ edges, we present an $\tilde{O}(m^{3/2})$ time \footnote{$\tilde{O}$ notation hides polylogarithmic factors.} $5/3$-approximation algorithm, and show that under SETH, neither the running time, nor the approximation factor can be significantly improved while keeping the other unchanged.\end{abstract}

\section{Introduction}
% ***bichromatic furthest pair well-studied in point sets
% ***diameter/furthest pair well-studied in graphs
% ***ST-version has been previously studied (stoc paper) but we show that surprisingly (?) much better approximation results are possible for the bichromatic version
% ***the ST-diameter algorithms didn't seem to work for directed graphs
% ***all of our results are tight 
% ***perhaps these new techniques developed for these problems offer a new way of looking these problems and could help in closing the gaps for regular diameter
% ***using previously known techniques fails for these problems and our algorithms and analyses require more delicate handling. 
% ***As an example of the previous point, algs for bichromatic diameter and directed bichromatic diameter. Describe why previous approaches don't work for these problems... The innovation for bichrom diam is exploitation of the fact that every vertex must be in either S or T to carefully choose an edge with one endpoint in S and the other in T etc. The innovation for directed bichrom diam is the observation that the set of vertices in the boundary that the true path uses is limited and using edge sampling of edges in the boundary etc.
A fundamental and very well studied problem in algorithms is the Diameter of a graph, where the output is the largest (shortest path) distance over all pairs of vertices. Over the years many different algorithms have been developed for the problem, both in theory (e.g. \cite{ACIM99,diam-prac1,RV13,ChechikLRSTW14,diamstoc18}) and in practice (e.g. \cite{diam-prac2,diam-prac4,diam-prac5}).

A very natural variant is the so called $ST$-Diameter problem \cite{diamstoc18}: given a graph and two subsets $S$ and $T$ of its vertex set, determine the largest distance between a vertex of $S$ and a vertex of $T$. In the Subset version of $ST$-Diameter, we have $S=T$.  Bichromatic Diameter is the version of $ST$-Diameter for which $S$ and $T$ partition the vertex set. Besides Diameter, the Radius (the smallest distance for which a ``center'' node can reach all other nodes) and Eccentricities (the largest distance out of every vertex) problems are also very well studied, and analogous $ST$, Subset, and Bichromatic versions are easy to define.

All of these parameters are simple to compute by computing all pairwise distances in the graph, i.e. by solving All-Pairs Shortest Paths (APSP). In sparse $n$-node graphs, where the number of edges $m$ is $\tilde{O}(n)$, APSP still needs $\Omega(n^2)$ time, as this is the size of the output, whereas it is apriori unclear whether this much time is needed for computing the Diameter, Radius and Eccentricities or their $ST$ and bichromatic variants, as the output is small.

%A slight variant is the {\em Bichromatic} Diameter problem: given a graph with a partitioning of its vertices into red and blue, determine the largest distance between a red vertex and a blue vertex. This variant is very well motivated (e.g. for finding the worst case distance between a supplier and a customer). 

A related extremely well-studied problem in computational geometry is Bichromatic Diameter on point sets (commonly known as Bichromatic Farthest Pair), where one seeks to determine the farthest pair of points in a given set of points in space (see e.g. \cite{yao-furthest,dumitrescu-furthest,ryan-furthest,agarwal-furthest,katoh}). Another related problem is the Subset version of spanners (e.g.~\cite{klein2006subset, cygan2013pairwise}), as well as the $ST$ version of spanners (e.g.~\cite{coppersmith2006sparse, kavitha2017new}). Furthermore, the $ST$, Subset, and Bichromatic versions of many problems have been of great interest; for instance Steiner Tree, Subset TSP, and a number of problems in computational geometry such as Bichromatic Matching (e.g.~\cite{indyk2007near}) and Bichromatic Line Segment Intersection (e.g.~\cite{chazelle1994algorithms}).

%The Bichromatic Diameter in {\em graphs}, however, is seemingly not as well studied in theory.

% A large reason for this is that with respect to exact computation, there are no non-trivial algorithms known for Diameter itself---the only known algorithm computes all the pairwise distances, i.e. solves All-Pairs Shortest Paths (APSP). Solving APSP, of course, also solves Bichromatic Diameter so in this sense there is no real difference between the two problems.\footnote{In point sets, it is not hard to show that Bichromatic Farthest Pair and Farthest Pair are equivalent. Thus, most of the literature is on Bichromatic Closest Pair instead.} Regarding \emph{approximability} however, we show a (conditional) separation between the two problems. %that the two problems are not equivalent In graphs however, we show (conditionally) that the Bichromatic Diameter problem is more difficult to \emph{approximate} than the regular Diameter problem, so in that sense the two problems are not equivalent. 

There are several known approximation algorithms for the standard version of Diameter, most of which have been developed in the last 6 years. Trivially, running Dijkstra's algorithm from an arbitrary vertex gives a simple $\tilde{O}(m)$ time $2$-approximation algorithm for directed and weighted graphs. Non-trivial algorithms achieve an improved approximation factor with an increased runtime: Building on Aingworth et al.~\cite{ACIM99}, Roditty and Vassilevska W.~\cite{RV13} showed for instance that an ``almost'' $1.5$ approximation for Diameter can be computed in $\tilde{O}(m\sqrt n)$ time in $m$-edge $n$-vertex directed weighted graphs---the approximation factor is $1.5$ if the Diameter is divisible by $3$, and there is a slight additive error otherwise. Chechik et al.~\cite{ChechikLRSTW14} gave a true $1.5$ approximation at the expense of increasing the runtime to $\tilde{O}(mn^{2/3})$, and Cairo, Grossi and Rizzi~\cite{cairo} generalized the approach giving an $\tilde{O}(mn^{1/(k+1)})$ time, ``almost'' $2-1/2^k$ approximation algorithm for all $k\geq 1$ which works only in undirected graphs. 

%Until recently, there were no non-trivial approximation algorithms for Bichromatic Diameter. 
In STOC'18, Backurs et al.~\cite{diamstoc18} gave the first non-trivial approximation algorithms for $ST$-Diameter: an $\tilde{O}(m^{3/2})$ time $2$-approximation and an $\tilde{O}(m)$ time $3$-approximation. They also showed that these algorithms cannot be improved significantly, unless the Strong Exponential Time Hypothesis (SETH) fails. Backurs et al. did not provide algorithms for $ST$-Eccentricities or $ST$-Radius, and they did not study the natural Subset and Bichromatic versions. They also only focused on undirected graphs.

We study the following natural and fundamental questions:
\begin{center}
{\em How well can $ST$-Eccentricities and $ST$-Radius be approximated? Are any interesting approximation algorithms possible for directed graphs for any of the $ST$-variants? Does the approximability of the problems change when one turns to the Subset versions in which $S=T$, or the Bichromatic versions in which $S$ and $T$ are required to partition the vertex set?}
\end{center}

\subsection{Our Results}
We present a comprehensive study of the approximability of the $ST$, Subset and Bichromatic variants of the Diameter, Radius and Eccentricities problems in graphs, both with and without directions and weights. We obtain the first non-trivial approximation algorithms for most of these problems, including time/accuracy trade-off upper and lower bounds. We show that nearly $\emph{all}$ of our approximation algorithms are tight under SETH (or under the related Hitting Set Hypothesis for Radius). Additionally, we study a parameterized version of these problems.

% Similarly, we present a comprehensive study of the approximability of the $ST$ variants of Radius and Eccentricities ($ST$-Diameter was resolved in~\cite{diamstoc18}). We present the first non-trivial approximation algorithms, which we also show are conditionally tight. 

% Furthermore, we provide a comprehensive study of the approximability of another natural special case of $ST$-Diameter: the case where $S=T$. We call this variant \emph{Subset} Diameter, and we also study Subset Radius and Eccentricities. Again, we present approximation algorithms and matching lower bounds.

%and obtain fast approximation algorithms whose bounds are tight under SETH and other fine-grained hypotheses. 

Our results are summarized in Tables~\ref{table:bichromundir}%, \ref{table:STundir}, \ref{table:bichromdir}, and 
-\ref{table:subset}. 

\begin{table}[h!]
\centering
    \begin{tabular}{|c||c|c|c||c|c|}
    \hline
      & \multicolumn{3}{c||}{Upper Bounds} &  \multicolumn{2}{c|}{Lower Bounds}  \\
      \hline
      Problem & Runtime  & Approx. &Comments& Runtime & Approx. \\ \hline \hline
      \multirow{4}{5em}{Diameter} & $O(m+n\log{n})$ & almost $2$ & unweighted, {\bf tight} & $m^{1+o(1)}$&$2-\delta$\\ \cline{2-6}
      &$\tilde{O}(m\sqrt{n})$ & almost $5/3$ & unweighted, nearly {\bf tight} &$m^{\frac{k}{k-1}-o(1)}$ & $2-\frac{1}{2k-1}-\delta$ \\ \cline{2-6}
      &$\tilde{O}(m^{3/2})$ & $5/3$ & weighted, {\bf tight}&" &"  \\ \cline{2-6}
      &$O(m|B|)$ & almost $3/2$ & unweighted, {\bf tight*}& $m^{2-o(1)}$ & $3/2-\delta$\\
      \hline
      \hline
       \multirow{4}{5em}{Radius} & $O(m+n\log{n})$ & almost $2$ & unweighted& &  \\ \cline{2-6}
      &$\tilde{O}(m\sqrt{n})$ & almost $5/3$ & unweighted, nearly {\bf tight*} &$m^{2-o(1)}$& $5/3 - \delta$  \\ \cline{2-6}
      &$\tilde{O}(m^{3/2})$ & $5/3$ & weighted, {\bf tight*}& "&"\\ \cline{2-6}
      &$O(m|B|)$ & almost $3/2$ &unweighted, {\bf tight*}& $m^{2-o(1)}$ & $3/2-\delta$\\\hline \hline
       \multirow{3}{5em}{Eccentricities} & $O(m+n\log{n})$ & $3$ & weighted, {\bf tight}&$m^{1+o(1)}$&$3-\delta$ \\ \cline{2-6}
      &$\tilde{O}(m\sqrt{n})$ & almost $2$ & unweighted, nearly {\bf tight}&$m^{\frac{k}{k-1}-o(1)}$ & $3-2/k-\delta$ \\ \cline{2-6}
      &$\tilde{O}(m^{3/2})$ & $2$ & weighted, {\bf tight}&" &"\\ \cline{2-6}
      &$O(m|B|)$ & almost $5/3$ &unweighted, {\bf tight*}& $m^{2-o(1)}$ & $5/3-\delta$\\
      \hline 
    \end{tabular}
    \caption[]{Bichromatic undirected results. All of our parameterized algorithms and near-linear time algorithms, except for directed Subset Radius and Eccentricities, are deterministic. The rest are randomized and work with high probability\footnotemark. Our lower bounds for Diameter and Eccentricities are under SETH and our lower bounds for Radius are under the Hitting Set (HS) Hypothesis, defined later. All of our lower bounds hold even for unweighted graphs. The trade-off lower bounds in terms of $k$ hold for any integer $k\geq 2$. $\delta$ is any constant $>0$. $B$ and $B'$ are parameters defined in our parameterized algorithms. The lower bound constructions for the parameterized algorithms have $|B|=\tilde{O}(1)$\\
    * Multiplicative approximation factor is tight, but not runtime.}
    %** For constant $k$, the lower bound does not hold for algorithms with additive error.}
    \label{table:bichromundir}
\end{table}
\footnotetext{\emph{with high probability} means with probability at least $1-1/n^c$ for all constants $c$.}
\begin{table}[h!]
\centering
    \begin{tabular}{|c||c|c|c||c|c|}
    \hline
      & \multicolumn{3}{c||}{Upper Bounds} &  \multicolumn{2}{c|}{Lower Bounds}  \\
      \hline
      Problem & Runtime  & Approx. &Comments& Runtime & Approx.  \\ \hline \hline
    Diameter & $\tilde{O}(m^{3/2})$ & $2$ & weighted, {\bf tight*} & $m^{2-o(1)}$ & $2-\delta$ \\ \cline{2-6}
     &$O(m|B'|)$ & almost $3/2$ & unweighted, {\bf tight*}& $m^{2-o(1)}$ & $3/2-\delta$\\
    \hline\hline
    Radius & N/A& N/A & weighted, {\bf tight} & $m^{2-o(1)}$ & any finite \\ \hline\hline
    Eccentricities & N/A & N/A & weighted, {\bf tight} & $m^{2-o(1)}$ & any finite \\ \hline
    \end{tabular}
    \caption{Bichromatic directed results. See caption of Table~\ref{table:bichromundir}. }
    \label{table:bichromdir}
\end{table}

\begin{table}[h!]
\centering
    \begin{tabular}{|c||c|c|c||c|c|c|}
    \hline
      & \multicolumn{3}{c||}{Upper Bounds} &  \multicolumn{2}{c|}{Lower Bounds}  \\
      \hline
      Problem & Runtime  & Approx. &Comments& Runtime & Approx.  \\ \hline \hline
      \multirow{3}{5em}{Diameter\cite{diamstoc18}} & $O(m+n\log{n})$ & $3$ & weighted, {\bf tight} &$m^{1+o(1)}$&$3-\delta$  \\ \cline{2-6}
      &$\tilde{O}(m\sqrt{n})$ & almost $2$ & unweighted, nearly {\bf tight} &$m^{\frac{k}{k-1}-o(1)}$ & $3-2/k-\delta$ 
       \\ \cline{2-6}
       &$\tilde{O}(m^{3/2})$ & $2$ & weighted, {\bf tight} & "&"\\ \hline \hline  
       \multirow{3}{5em}{Radius} &  $O(m+n\log{n})$ & $3$ & weighted&  &   \\ \cline{2-6}
      &$\tilde{O}(m\sqrt{n})$ & almost $2$ & unweighted, nearly {\bf tight*}& $m^{2-o(1)}$ & $2 - \delta$\\ \cline{2-6}
      &$\tilde{O}(m^{3/2})$ & $2$ & weighted, {\bf tight*}&" &"  \\ \hline \hline
       \multirow{3}{5em}{Eccentricities} & $O(m+n\log{n})$ & $3$ & weighted, {\bf tight}&$m^{1+o(1)}$&$3-\delta$~\cite{diamstoc18}   \\ \cline{2-6}
      &$\tilde{O}(m\sqrt{n})$ & almost $2$ & unweighted, nearly {\bf tight}&$m^{\frac{k}{k-1}-o(1)}$ & $3-2/k-\epsilon$~\cite{diamstoc18}  \\ \cline{2-6}
      &$\tilde{O}(m^{3/2})$ & $2$ & weighted, {\bf tight}& "&" \\ \hline 
      
    \end{tabular}
    \caption{ST undirected results. See caption of Table~\ref{table:bichromundir}.}
    \label{table:STundir}
\end{table}

\begin{table}[h!]
\centering
    \begin{tabular}{|c||c|c|c||c|c|}
    \hline
      & \multicolumn{3}{c||}{Upper Bounds} &  \multicolumn{2}{c|}{Lower Bounds}  \\
      \hline
      Problem & Runtime  & Approx. &Comments& Runtime & Approx.  \\ \hline \hline
    Diameter & $\tilde{O}(m)$ & $2$ & weighted, directed, {\bf tight} & $m^{2-o(1)}$ & $2-\delta$ \\ \cline{2-6}
    \hline\hline
    \multirow{2}{5em}{Radius} &$\tilde{O}(m)$ & $2$ & weighted, undirected, {\bf tight} & $m^{2-o(1)}$ & $2-\delta$ \\ \cline{2-6}
      &$\tilde{O}(m/\delta)$ & $2+\delta$ & weighted, directed, {\bf tight} up to an additive $\delta$ &"&" \\ \hline \hline
    Eccentricities &$\tilde{O}(m/\delta)$ & $2+\delta$ &weighted, directed, {\bf tight} up to an additive $\delta$ & $m^{2-o(1)}$ & $2-\delta$\\ \hline
    \end{tabular}
    \caption{Subset results. See caption of Table~\ref{table:bichromundir}. }
    \label{table:subset}
\end{table}

All our algorithms in $m$-edge, $n$-node graphs, run in  $\tilde{O}(m^{3/2})$ time or in $\tilde{O}(m\sqrt n)$ time when a small additive error is allowed. For sparse graphs the $m^{3/2}$ runtime beats the fastest APSP algorithms \cite{chan07,PettieR05,Pettie04} as they run in $\tilde{O}(mn)$ time. The $m\sqrt n$ time of the algorithms that allow small additive error beat the APSP algorithms for every graph sparsity.
\paragraph*{Bichromatic Diameter and Radius.}
Our first contribution is an algorithm with the same running time as the $2$-approximation $ST$-Diameter algorithm of \cite{diamstoc18},  achieving a better, $5/3$ approximation for Bichromatic Diameter. In other words, when $S$ and $T$ partition the vertex set of the graph, $ST$-Diameter can be approximated much better! Moreover, we show that under SETH, neither the runtime nor the approximation factor of our algorithm can be improved. The result is summarized in Theorem~\ref{thm:bichromfull} below, and proven in Theorems~\ref{thm:bichromedge} and \ref{thm:bidiamlb}.

\begin{theorem}\label{thm:bichromfull}
There is a randomized $\tilde{O}(m^{3/2})$ time algorithm, that given an undirected graph $G=(V,E)$ with nonnegative integer edge weights and $S\subseteq V, T=V\setminus S$, can output an estimate $D'$ such that $3D_{ST}/5\leq D'\leq D_{ST}$ with high probability, where $D_{ST}$ is the $ST$-Diameter of $G$. 

Moreover, if there is an $O(m^{3/2-\eps})$ time $5/3$-approximation algorithm for some $\eps>0$, or if there is an $O(m^{2-\eps})$ time $(5/3-\eps)$-approximation algorithm for the problem, then SETH is false.
\end{theorem}

We also obtain an $\tilde{O}(m\sqrt n)$ time algorithm that achieves an ``almost'' $5/3$-approximation: the guarantee for unweighted graphs is $3D_{ST}/5-6/5\leq D'\leq D_{ST}$.
We also obtain a near-linear time algorithm for weighted graphs that returns an estimate $D'$ with $D_{ST}/2-W/2\leq D'\leq D_{ST}$ where $W$ is the minimum weight of a $S\times T$ edge. Using our general theorem \ref{thm:bidiamlb}, we get that this result is also essentially tight, as a $(2-\eps)$-approximation for $\eps>0$ running in near-linear time would refute SETH.

To obtain our improvements for Bichromatic Diameter over the known $ST$-Diameter algorithms, we crucially exploit the basic fact that as $S, T$ partition $V$ any path that starts from a vertex $s \in S$ and ends in a vertex $t \in T$ must cross a $(u, v)$ edge such that $u \in S, v \in T$. While this fact is clear, it not at all obvious how one might try to exploit it.

We explain our technique in more detail for the bichromatic diameter problem, and similar ideas are used for our algorithms for the other problems. Let $s^*\in S$ and $t^*\in T$ be end-points of an $ST$-Diameter path. Similarly to prior Diameter algorithms, our goal is to run Dijkstra's algorithm from some $s\in S$ which is close to $s^*$, and hence far from $t^*$, or from some $t\in T$ which is close to $t^*$ and hence far from $s^*$ (by the triangle inequality).
Our $5/3$-approximation algorithms are a delicate combination of two themes: (1) randomly sample nodes in $S$ and nodes in $T$ -- similarly to prior works, the sampling works well if there are many nodes of $S$ that are close to $s^*$, or if there are many nodes of $T$ that are close to $t^*$. If (1) is not good enough, in theme (2) we show that we can find a node $w\in S$ close to $t^*$ for which we can ``catch'' an $S\times T$ edge $(s,t)$ on the shortest $w\rightarrow t^*$ path, such that $t$ is close to $t^*$. Theme (2)  is our new contribution. Because of theme (2), our algorithms are more complicated than the $ST$-Diameter algorithms, but run in asymptotically the same time, and achieve a better approximation guarantee. 
In order to better separate the ideas in our algorithms, we explain them in several steps, where Theme (1) can be seen in the first steps and Theme (2) appears towards the last steps.
%Our algorithms consist of several steps, where in order to better see the distinction of Theme (1) and Theme (2) in our algorithms, 

Following a similar approach to our Bichromatic Diameter algorithms, we develop similar algorithms for Bichromatic Radius. First, we give a simple near-linear time almost $2$-approximation algorithm, and then we adapt the $5/3$-approximation for Bichromatic Diameter to also give a $5/3$-approximation for Bichromatic Radius. Moreover, we show that any better approximation factor requires essentially quadratic time, under the Hitting Set (HS) Hypothesis of \cite{AbboudWW16} (see also \cite{GaoIKW17}).
\begin{theorem}
\label{ref:rad}
There is a randomized $\tilde{O}(m^{3/2})$ time algorithm, that given an undirected graph $G=(V,E)$ with nonnegative integer edge weights and $S\subseteq V, T=V\setminus S$, can output an estimate $R'$ such that $R_{ST}\leq R'\leq 5R_{ST}/3$ with high probability, where $R_{ST}$ is the $ST$-Radius of $G$. Moreover, if there is a $5/3-\eps$ approximation algorithm running in $O(m^{2-\delta})$ time for any $\eps,\delta>0$, then the HS Hypothesis is false.
\end{theorem}

Similarly to the Bichromatic Diameter algorithm, if one is satisfied with a slight additive error, one can improve the runtime to $\tilde{O}(m\sqrt n)$.

\paragraph*{$ST$-Eccentricities and $ST$-Radius.}
Prior work only considered $ST$-Diameter but did not consider the more general $ST$-Eccentricities problem in which one wants to approximate for every $s\in S$, $\eps_{ST}(s):=\max_{t\in T} d(s,t)$.

Here we show that one can achieve exactly the same approximation factors for $ST$-Eccentricities as for $ST$-Diameter. Since any conditional lower bound for $ST$-Diameter also applies for the $ST$-Eccentricities problem, the algorithms we obtain are conditionally optimal, similarly to the $ST$-Diameter algorithms in \cite{diamstoc18}. Interestingly, we show that the same conditional lower bounds apply for Bichromatic Eccentricities (Proposition~\ref{prop:whoknows}), and therefore our $ST$-Eccentricities algorithms are optimal even for the Bichromatic case.

\begin{theorem}
There is a randomized $\tilde{O}(m^{3/2})$ time algorithm, that given an undirected graph $G=(V,E)$ with nonnegative integer edge weights and $S, T\subseteq V$, can output for every $s\in S$, an estimate $\eps'(s)$ such that $\eps_{ST}(s)/2\leq \eps'(s)\leq \eps_{ST}(s)$ with high probability.
Moreover, if there is a $2-\eps$ approximation algorithm running in $O(m^{2-\delta})$ time for any $\eps,\delta>0$ or a $2$-approximation algorithm running in $O(m^{3/2-\eps})$ time for $\eps>0$, even for the Bichromatic case when $T=V\setminus S$, then SETH is false.
\end{theorem}

Again, as before, one can improve the runtime to $\tilde{O}(m\sqrt n)$ with a slight additive error, and there is a simple near-linear time $3$-approximation algorithm which is tight under SETH, similar to the one in \cite{diamstoc18} for $ST$-Diameter.
A simple argument shows that these algorithms imply algorithms with the same running time and approximation factor for $ST$-Radius.

\paragraph*{Bichromatic and $ST$ Problems in Directed Graphs.}
Using simple reductions we first show that there can be no $O(m^{2-\eps})$ time (for $\eps>0$) algorithms that achieve any finite approximation for $ST$-Diameter or $ST$-Eccentricities (under SETH), or $ST$-Radius (under HS).
Interestingly, the same holds for Bichromatic Eccentricities (under SETH, Proposition~\ref{prop:bidirecclb}) and Bichromatic Radius (under HS, Proposition~\ref{prop:bidirradlb}), but not Bichromatic Diameter!
Surprisingly, unlike those two problems, Bichromatic Diameter does admit a finite, in fact $2$-approximation algorithm running in subquadratic time, and this algorithm is conditionally optimal:

\begin{theorem}
There is a randomized $\tilde{O}(m^{3/2})$ time algorithm, that given a directed graph $G=(V,E)$ with nonnegative integer edge weights and $S\subseteq V, T=V\setminus S$, can output an estimate $D'$ such that $D_{ST}/2\leq D'\leq D_{ST}$ with high probability, where $D_{ST}$ is the $ST$-Diameter of $G$. 

Moreover, if there is an $O(m^{2-\eps})$ time $2-\delta$-approximation algorithm for the problem for some $\eps,\delta>0$, then SETH is false.
\end{theorem}

The previously known techniques for approximating Diameter in directed graphs fail here. The main issue is that the prior techniques were general enough that they also gave algorithms for Eccentricities and Radius as a byproduct. In the Bichromatic case, however, there is a genuine difference between Diameter and Radius, as we noted above, and new techniques are needed. Here again it turns out that combining theme (2) with a delicate argument is sufficient to get conditionally tight algorithms under SETH. 

% %get rid of mention of $B$ not in boundary algs.
% We consider a randomly sampled set of edges $B$ in $S\times T$ and base our algorithm on two cases. Letting $s^*\in S$ and $t^*$ be the endpoints of the $ST$-Diameter path, the first case is that there exists a $(s,t)\in B$ for which $d(s^*,s)\leq D_{ST}/2$. Then $\max_{t'\in T} d(s,t')\geq D_{ST}/2$ and the quantity is also at most $D_{ST}$. Otherwise, there is some node of $S$ very far from the vertices in $S$ which are incident to $B$, let $w\in S$ be furthest such vertex. Then we show either $w$ serves as a good replacement for $s^*$ i.e $\max_{t'\in T} d(w, t'))\geq D_{ST}/2$ or that the edges close to $w$ allow us to ``catch'' an $S\times T$ edge $(s,t)$ on the shortest $w\rightarrow t^*$ path. And then we show that if none of the previous estimates work, $t$ would be close to $t^*$ and hence serve as a good replacement for $t^*$ i.e. $\max_{s'\in S} d(s',t))\geq D_{ST}/2$.
\paragraph*{Subset Versions.}
Recall that Subset Diameter, Radius, and Eccentricities are the versions of the corresponding $ST$ problems with the constraint that $S=T$. Interestingly, Subset Diameter, Radius, and Eccentricities all exhibit the same sharp threshold behavior. For all three problems, there are near-linear time algorithms that achieve a 2 (or almost 2) approximation, as well as conditional lower bounds that show that there is no $2-\delta$ approximation in $m^{2-o(1)}$ time.

\paragraph*{Parameterized Algorithms.} We consider the Bichromatic Diameter, Radius, and Eccentricities problems parameterized by the size of the \emph{boundary} between the $S$ and $T$ sets. If $S'$ is the set of vertices in $S$ that have a neighbor in $T$, and $T'$ is the set of vertices in $T$ that have a neighbor in $S$, then the boundary $B$ is whichever of $S'$ or $T'$ is smaller in size. Our lower bound constructions already have small boundary so they rule out algorithms even for graphs with small boundary. However, interestingly we obtain near-linear time algorithms for graphs with small boundary that achieve \emph{better} multiplicative approximation factors than the optimal non-parameterized algorithms. This is not a contradiction because our parameterized algorithms have a constant additive error, while the apparently contradictory lower bounds do not tolerate additive error.

% We consider algorithms and conditional lower bounds for special
% cases. One such case is when the instance of Bichromatic Diameter has a small number of edges in $S\times T$. We obtain a better approximation algorithm for this case, and also show that it is tight.
% We also consider a special case of $ST$-Diameter, Radius and Eccentricities when $S=T$. We call these the Subset versions. We obtain simple near-linear time $2$-approximations by slight modifications of prior work.

\section{Preliminaries}
Given a graph $G=(V,E)$ (directed or undirected, weighted or unweighted), let $d(u,v)$ denote the distance from $u\in V$ to $v\in V$. For a subset $X\subseteq V$ and $v\in V$, define $d(v,X):=\min_{x\in X} d(v,x)$. Similarly $d(X,v):=\min_{x\in X} d(x,v)$.

Unless otherwise stated, $m$ denotes the number of edges and $n$ the number of vertices of the underlying graph. Without loss of generality, we can assume that all undirected graphs are connected, and all directed graphs are weakly connected, so that $m\geq n-1$.

The {\em Eccentricity} $\eps(v)$ of a vertex $v\in V$ is $\max_{u\in V} d(v,u)$. The {\em Diameter} $D(G)$ of $G$ is $\max_{v\in V}\eps(v)$, and the {\em Radius} $R(G)$ of $G$ is $\min_{v\in V} \eps(v)$.

Given $S,T\subseteq V$, we define analogous parameters as follows. The $ST$-Eccentricity $\eps_{ST}(v)$ of $v\in S$ is $\max_{u\in T} d(v,u)$. The $ST$-Diameter $D_{ST}(G)$ is $\max_{v\in S} \eps_{ST}(v)$, and the $ST$-Radius $R_{ST}(G)$ is $\min_{v\in S}\eps_{ST}(v)$. 
%The symmetric $ST$-radius $\tilde{R}_{ST}(G)$ of $G$ is $\max \{R_{ST}(G),R_{TS}(G)\}$. THE GUARANTEES ARE NOT AS I THOUGHT SO I AM REMOVING SYMRADIUS

The above parameters are called {\em Bichromatic} Eccentricities, Diameter, and Radius if $S$ and $T$ form a partition of $V$, i.e. $T=V\setminus S$.

The above parameters are called {\em Subset} Eccentricities, Diameter, and Radius if $S=T$ and are notated with subscript $S$ instead of $ST$.

\subsection{Preliminaries for algorithms}
\begin{lemma}\label{lemma:disthit}
Let $G=(V,E)$ be a (possibly directed and weighted graph) and let $W\subseteq V$. Let $g\geq \Omega(\ln n)$ be an integer. Let $S\subseteq W$ be a random subset of $c (|W|/g) \ln n$ vertices for some constant $c>1$. For every $v\in V$, let $W(v)$ be the set of vertices $x\in W$ for which $d(v,x)<d(v,S)$. Then with probability at least $1-1/n^{c-1}$, for every $v\in V$, $|W(v)|\leq g$, and moreover, if one takes the closest $g$ vertices of $W$ to $v$, they will contain $W(v)$.
\end{lemma}

\begin{proof}
For each $v\in V$, imagine sorting the nodes $x\in W$ according to $d(v,x)$.
Define $Q_v$ to be the first $g$ nodes in this sorted order - those are the nodes of $W$ closest to $v$ (in the $v\rightarrow x$ direction).

We pick $S$ randomly by selecting each vertex of $W$ with probability $(c\ln n)/g$.
The probability that a particular $q\in Q_v$ is not in $S$ is $1-(c\ln n)/g$, and the probability that no $q\in Q_v$ is in $S$ is $(1-(c\ln n)/g)^g \leq 1/n^c$. By a union bound, with probability at least $1-1/n^{c-1}$, for every $v\in V$, we have that $Q_v\cap S\neq \emptyset$.

Now, for each particular $v$, say that $w(v)$ is a node in $Q_v\cap S$. Since all nodes $x\in W$ with $d(v,x)<d(v,w(v))$ must be in $Q_v$, and since $d(v,w(v))\geq d(v,S)$, we must have that $W(v)\subseteq Q_v$. Hence, with probability at least $1-1/n^{c-1}$, for every $v\in V$, $|W(v)|\leq g$ and $W(v)\subseteq Q_v$.
\end{proof}

\begin{lemma}\label{lemma:farpoint}
Let $G=(V,E)$ be a (possibly directed and weighted) graph. Let $M,W\subseteq V$ and 
%, and let $v\in M$. 
let $S\subseteq W$ be a random subset of $c (n/g) \ln n$ vertices for some large enough constant $c$ and some integer $g\geq 1$.
%Let $w\in M$ be the node of $M$ furthest from $S$ (i.e. maximizing $d(w,S)$). 

Then, for any $D>0$ and
%either (1) there is some $s\in S$ with $d(v,s)\leq D$, or (2)
for any $w\in M$ with $d(w,S)>D$, if one takes the closest $g$ vertices of $W$ to $w$, they will contain all nodes of $W$ at distance $<D$ from $w$, with high probability.
\end{lemma}

\begin{proof}
%Suppose that there is no $s\in S$ with $d(v,s)\leq D$. Thus, $d(v,S)>D$. Since $v\in M$ and $w$ is the furthest in $M$ from $S$, $d(w,S)>D$. 
Let $Q$ be the closest $g$ vertices of $W$ to $w$.
By Lemma~\ref{lemma:disthit}, with high probability $Q$ contains all nodes of $W$ at distance $<d(w,S)$ from $w$, and hence $Q$ contains all nodes of $W$ at distance $<D$ from $w$, with high probability.
\end{proof}

We sometimes sample edges instead of vertices, so analogous lemmas to Lemmas~\ref{lemma:disthit} and \ref{lemma:farpoint} hold when the sample is from a set of edges. Here is the analogue of Lemma~\ref{lemma:farpoint}. The other lemma is similar.

\begin{lemma}\label{lemma:edgefarpoint}
Let $G=(V,E)$ be a (possibly directed and weighted graph) and let $M,W\subseteq V$.
%, and let $x\in M$.
Let $E'\subseteq E$ be a random subset of $c (|E|/g) \ln n$ edges for some large enough constant $c$ and some integer $g\geq 1$.
Let $Q$ be the endpoints of edges in $E'$ that are in $W$.
%Let $w\in M$ be the node of $M$ furthest from $Q$.

Then, for any $D>0$, 
%(1) either there is some $s\in Q$ and $d(x,s)\leq D$, or (2)
and for any $w$ with $d(w,S)>D$, if one takes the closest $g$ edges of $E'$ to $w$ wrt the distance from their $W$ endpoints, they will contain all edges of $E'$ whose $W$ endpoints are at distance $<D$ from $w$, with high probability.
\end{lemma}
%[[Use of $W_2$? Do we want edges with one vertex in $W_2$?]]

% \begin{lemma}\label{lemma:edgedisthit}
% Let $G=(V,E)$ be a (possibly directed and weighted graph). Let $W_1,W_2\subseteq V$, and
% let $W\subseteq E\cap (W_1\cup W_2)$. Let $g\geq \Omega(\ln n)$ be an integer. Let $S\subseteq W$ be a random subset of $c (|W|/g) \ln n$ vertices for some constant $c>1$. For every $v\in V$, let $W(v)$ be the set of edges $(x,y)\in W$ for which $x\in W_1,y\in W_2$ and $d(v,x)<d(v,S)$. Then with probability at least $1-1/n^{c-1}$, for every $v\in V$, we have $|W(v)|\leq g$, and moreover, if one takes the closest $g$ edges of $W$ to $v$ (with respect to their $W_1$ endpoint), they will contain $W(v)$.
% \end{lemma}

\subsection{Preliminaries for lower bounds}
The Strong Exponential Time Hypothesis (SETH) asserts that on a Word-RAM with $O(\log n)$ bit words, there is no $(2-\eps)^n$ time (possibly randomized) algorithm for some constant $\eps > 0$ that can determine whether a given CNF-Formula with $n$ variables and $O(n)$ clauses is satisfiable. (This version of SETH is equivalent to the original formulation by Impagliazzo, Paturi and Zane~\cite{ipz1}.) By a result of Williams~\cite{TCS05}, the following Orthogonal Vectors (OV) Problem requires $n^{2-o(1)}\poly(d)$ time (on a word-RAM with $O(\log n)$ bit words), unless SETH fails: given two sets $U,V\subseteq \{0,1\}^d$ with $|U|=|V|=n$ and $d=\omega(\log n)$, determine whether there are $u\in U,v\in V$ with $u\cdot v=0$. %We refer to this as the OV Hypothesis. 

Given an arbitrary instance of OV with $d=\tilde{O}(1)$ (while respecting $d=\omega(\log n)$, e.g. $d=\Theta(\log^2 n)$), consider the following graph representation, which we call the {\em OV-graph}: the vertex set consists of a node for every $u\in U$, for every $v\in V$ and for every coordinate $c\in [d]=C$, and there is an edge $(x\in U\cup V,c\in C)$ if and only if $x[c]=1$. OV is then equivalent to the question of whether there exist $u\in U,v\in V$ such that $d(u,v)>2$. In fact, it is equivalent to distinguishing whether for every $u\in U,v\in V$, $d(u,v)=2$ (no OV-solution), or there is some $u\in U,v\in V$ such that $d(u,v)\geq 4$ (OV-solution). In other words, if we set $S=U,T=V$, the $ST$-Diameter of the OV-graph is $2$ if and only if there is no OV-solution and at least $4$ otherwise. Because the OV graph has $m=\tilde{O}(n)$, under SETH, any $(2-\delta)$-approximation algorithm for $ST$-Diameter requires $m^{2-o(1)}$.

A related problem to OV is the Hitting Set (HS) problem~\cite{AbboudWW16,GaoIKW17,ipecsurvey}: given two sets $U,V\subseteq \{0,1\}^d$ with $|U|=|V|=n$ and $d=\omega(\log n)$, determine whether there is $u\in U$ such that for all $v\in V$, $u\cdot v\neq 0$. 
A common hypothesis is that (on the word-RAM) HS requires $n^{2-o(1)}$ time.

If we form the OV-graph on the HS instance input, then the HS problem becomes equivalent to determining whether there is some $u\in U$ such that for all $v\in V$, $d(u,v)\leq 2$. In other words, if we set $S=U,T=V$, the $ST$-Radius of the OV-graph is $2$ if and only if there is a HS-solution and at least $4$ otherwise. Thus, under the HS hypothesis, any $(2-\delta)$-approximation algorithm for $ST$-Radius requires $m^{2-o(1)}$.

Additionally for our constructions we assume that if there is a HS solution $u'$ then for all $c \in C$, $d(u', c) \leq 3$. This is because for every coordinate index $i$ there must be $v \in V$ with $v[i] = 1$ as otherwise we can just delete the $i^{th}$ bit from all vectors.  

Let $k\geq 2$ be an integer. Then, a generalization of the OV problem is $k$-OV: given $k$ sets $U_1,\ldots, U_k\subseteq \{0,1\}^d$, are there $u_1\in U_1,\ldots, u_k\in U_k$ so that $\sum_{c=1}^d\prod_{i=1}^k u_i[c] =0$? It is known that, under SETH, when $d=\omega(\log n)$, there is no $n^{k-o(1)}$ time algorithm for $k$-OV (in the word RAM model)~\cite{TCS05}.

Similar to the OV-graph, Backurs et al.~\cite{diamstoc18} define a graph for $k$-OV which we will refer to as the $k$-OV-graph. We do not explicitly define the $k$-OV-graph here; instead we list its properties in the following theorem.

\begin{theorem}[\cite{diamstoc18}] \label{thm:kOV} Let $k\geq 2$.
	Given a $k$-OV instance consisting of sets $W_0,W_1,\dots,W_{k-1} \subseteq \{0,1\}^d$, each of size $n$, we can in $O(kn^{k-1}d^{k-1})$ time construct an unweighted, undirected graph with %\newline 
	$O(n^{k-1}+k n^{k-2} d^{k-1})$ vertices and $O(k n^{k-1} d^{k-1})$ edges that satisfies the following properties.
	\begin{enumerate}[itemsep=0mm]
		\item The graph consists of $k+1$ layers of vertices $L_0,L_1,L_2,\dots,\allowbreak L_k$. The number of nodes in the sets is $|L_0|=|L_k|=n^{k-1}$ and $|L_1|,|L_2|,\dots,|L_{k-1}|\leq n^{k-2}d^{k-1}$.
		\item $L_0$ consists of all tuples $(a_0,a_1,\ldots, a_{k-2})$ where for each $i$, $a_i\in W_i$. Similarly, $L_k$ consists of all tuples $(b_1,b_2,\ldots, b_{k-1})$ where for each $i$, $b_i\in W_i$.
		\item If the $k$-OV instance has no solution, then $d(u,v)=k$ for all $u \in L_0$ and $v \in L_k$.
		\item If the $k$-OV instance has a solution $a_0, a_1,\dots,a_{k-1}$ where for each $i$, $a_i\in W_i$ then if $\alpha=(a_0,\dots a_{k-2})\in L_0$ and $\beta= (a_1,\dots,a_{k-1}) \in L_k$, then $d(\alpha,\beta)\geq 3k-2$. 
% 		\item Suppose the $k$-OV instance has a solution $a_0, a_1,\dots,a_{k-1}$ where for each $i$, $a_i\in W_i$. Let $t=k-2$.
% 		% then for any tuple $(b_{\lfloor\frac{k-1}{2}\rfloor},\dots,b_{k-2})$, if $\alpha=(a_0,a_1,\ldots,a_{\lfloor\frac{k-3}{2}\rfloor},b_{\lfloor\frac{k-1}{2}\rfloor},\ldots,b_{k-2})\in S$ and $\beta=(a_1,\dots,a_{k-1})\in T$, then $d(\alpha,\beta)\geq 3k-2-2\lceil\frac{k-1}{2}\rceil$.
% 		Let $s$ be such that $0\leq s\leq t$.

% Let $b_{t-s+j}\in W_{t-s+j}$ for all $j\in [1,\ldots,s]$ be some other vectors, potentially different from $a_{t-s+j}$.  
% Consider $\alpha=(a_0,a_1,\ldots,a_{t-s},\allowbreak b_{t-s+1},\ldots,b_{t})\in L_0$ and $\beta=(a_{1},\ldots,a_{t+1})\in L_{t+2}$.
% Then the distance between $\alpha$ and $\beta$ is at least $3t-2s+4$.

% Symmetrically, let $c_{j}\in W_{j}$ for all $j\in [1,\ldots,s]$ be some other vectors, potentially different from $a_{j}$.  
% Consider $\alpha=(a_0,a_1,\ldots,a_{t})\in L_0$ and $\beta=(c_{1},\ldots,c_s,a_{s+1},\dots, a_{t+1})\in L_{t+2}$.
% Then the distance between $\alpha$ and $\beta$ is at least $3t-2s+4$.
		\item For all $i$ from 1 to $k-1$, for all $v \in L_i$ there exists a vertex in $L_{i-1}$ adjacent to $v$ and a vertex in $L_{i+1}$ adjacent to $v$. 
	\end{enumerate}
\end{theorem}

\subsection{Organization}
In Section~\ref{sec:bialgs} we present our algorithms for Bichromatic Diameter, Eccentricities, and Radius. In Section~\ref{sec:stalgs} we present our algorithms for $ST$-Eccentricities and Radius. In Section~\ref{sec:subalgs} we present our algorithms for Subset Diameter, Eccentricities, and Radius. In Section~\ref{sec:paramalgs} we present our parameterized algorithms for Bichromatic Diameter, Radius, and Eccentricities. In Section~\ref{sec:lbs} we present all of our conditional lower bounds.

\section{Algorithms for Undirected  Bichromatic Diameter, Eccentricities and Radius
}\label{sec:bialgs}

\subsection{Undirected Bichromatic Diameter}
We begin with a simple near-linear time algorithm. 

\begin{proposition}\label{prop:linbidiam}
There is an ${O}(m+n\log n)$ time algorithm, that given an undirected graph $G=(V,E)$ and $S\subseteq V, T=V\setminus S$, can output an estimate $D'$ such that $D_{ST}(G)/2-W/2 \leq D'\leq D_{ST},$ where $W$ is the minimum weight of an edge in $S\times T$.
\end{proposition}

\begin{proof}
Let $(s,t)$ be a minimum weight edge of $G$ with $s\in S$ and $t\in T$. Run Dijkstra's algorithm from $s$ and from $t$. Let $D'=\max\{\max_{t'\in T} d(s,t'),\max_{s'\in S} d(s',t)\}$. Let $s^*\in S, t^*\in T$ be endpoints of an $ST$-Diameter path, i.e. $d(s^*,t^*)=D_{ST}$. Then, suppose that $\max_{t'\in T} d(s,t')< D_{ST}/2-W/2$. In particular, $d(s,t^*)<D_{ST}/2-W/2$, and hence $d(s,s^*)>D_{ST}/2+W/2$ by the triangle inequality. Also by the triangle inequality, $$D_{ST}/2+W/2<d(s,t)+d(t,s^*)\leq w(s,t)+ \max_{s'\in S} d(s',t).$$ Hence, 
$D'>D_{ST}/2-W/2,$ where $W$ is the minimum weight of an edge in $S\times T$.
\end{proof}

% \begin{proposition}
% (Previous version) There is an ${O}(m+n\log n)$ time algorithm, that given an undirected graph $G=(V,E)$ and $S\subseteq V, T=V\setminus S$, can output an estimate $D'$ such that $D_{ST}(G)/2-W< D'\leq D_{ST},$ where $W$ is the minimum weight of an edge in $S\times T$.
% \end{proposition}

% \begin{proof}
% Let $(s,t)$ be minimum weight edge of $G$ with $s\in S$ and $t\in T$. Run Dijkstra's algorithm from $s$ and from $t$. Let $D'=\max\{\max_{t'\in T} d(s,t'),\max_{s'\in S} d(s',t)\}$. Let $s^*\in S, t^*\in T$ be endpoints of an $ST$-Diameter path, i.e. $d(s^*,t^*)=D_{ST}$. Then, suppose that $\max_{t'\in T} d(s,t')< D_{ST}/2$. In particular, $d(s,t^*)<D_{ST}/2$, and hence $d(s,s^*)>D_{ST}/2$ by the triangle inequality. Also by the triangle inequality, $$D_{ST}/2<d(s,t)+d(t,s^*)\leq w(s,t)+ \max_{s'\in S} d(s',t).$$ Hence, 
% $D'>D_{ST}/2-W,$ where $W$ is the minimum weight of an edge in $S\times T$.
% \end{proof}

Now we turn to our $5/3$-approximation algorithms.
Our first theorem is for unweighted graphs. Later on, we modify the algorithm in this theorem to obtain an algorithm for weighted graphs as well, and at the same time remove the small additive error that appears in the theorem below.

\begin{theorem}\label{}
There is an $\tilde{O}(m\sqrt n)$ time algorithm, that given an unweighted undirected graph $G=(V,E)$ and $S\subseteq V, T=V\setminus S$, can output an estimate $D'$ such that $3D_{ST}(G)/5\leq D'\leq D_{ST}(G)$ if $D_{ST}(G)$ is divisible by $5$, and otherwise $3D_{ST}(G)/5-6/5\leq D'\leq D_{ST}(G)$. 
\end{theorem}

\begin{proof}
Let $D=D_{ST}(G)$ and let us assume that $D$ is divisible by $5$. If $D$ is not divisible by $5$, the estimate we return will have a small additive error. For clarity of presentation, we omit the analysis of the case where $D$ is not divisible by 5. However, we include such analyses in our proofs for Bichromatic Radius (Theorem~\ref{thm:bichromedgeradius}) and $ST$-Eccentricities (Theorem~\ref{STecc2approx}) and the analysis for Diameter is analogous.
%D=5d+e

Suppose the (bichromatic) $ST$-Diameter endpoints are $s^*\in S$ and $t^*\in T$ and that the $ST$-Diameter is $D$. The algorithm does not know $D$, but we will use it in the analysis. 

{\bf (Algorithm Step 1):}
 The algorithm first samples $Z\subseteq S$ of size $c \sqrt{n} \ln n$ uniformly at random. For every $z\in Z$, run BFS, and let $D_1 = \max_{z\in Z,t\in T} d(z,t)$.

{\bf (Analysis Step 1):} If for some $s'\in Z$ we have that $d(s^*,s')\leq 2D/5$, then $D_1\geq d(s',t^*)\geq D -  d(s^*,s') \geq 3D/5$. 
%Let us assume that $d(s^*,Z)>2D/5$.

%This will happen with high probability if the ball of $S$-nodes at distance $\leq 2D/5$ around $s^*$ has at least $\sqrt n$ nodes. Otherwise, we can assume that there are fewer than $\sqrt n$ $S$-nodes at distance $\leq 2D/5$ from $s^*$.
%if $d(s^*,s')\leq 2d+y$, $D_1\geq 3d+e-y$, so fewer ... dist $\leq 2d+y$.

{\bf (Algorithm Step 2):} Now, sample a set $X$ from $T$ of size $C \sqrt n \ln n$ uniformly at random for large enough constant $C$.
For every $t \in X$, run BFS and find the closest node $s(t)$ of $S$ to $t$.
Run BFS from every $s(t)$. Let $D_2=\max_{t\in X, t'\in T} d(s(t),t')$.

{\bf (Analysis Step 2):} If $s^*$ is at distance $\leq D/5$ from some node $t$ of $X$, then $d(s^*,s(t))\leq 2D/5$ (since $s(t)$ is closer to $t$ than $s^*$), and so $D_2\geq d(s(t),t^*)\geq 3D/5$. 
%D/5 -> d+y/2, get 2d+e-y.

If neither $D_1$, nor $D_2$ are good approximations, it must be that $d(s^*,X)>D/5$ and $d(s^*,Z)>2D/5$. 
Consider the nodes $M$ of $S$ that are at distance $> 2D/5$ from $Z$, 
then the node $w\in M$ that is furthest from $X$ among all nodes of $M$.
If neither $D_1$, nor $D_2$ was a good approximation, $s^*\in M$ and since $d(s^*,X)>D/5$, we must have that $d(w,X)>D/5$ (and also $d(w,Z)>2D/5$). In the next step we will look for such a $w$.
% , would have $\leq \sqrt n$ nodes of $T$ at distance $\leq D/5$ from it and also $\leq \sqrt n$ nodes of $S$ at distance $\leq 2D/5$ from it, by Lemma~\ref{lemma:farpoint}. 
% 2d+y and (2d+y)/2

{\bf (Algorithm Step 3):} For each $s \in S$ define $D_s$ to be the biggest integer which satisfies $d(s,X)>D_s/5$ and $d(s,Z)>2D_s/5$. Let $w = \arg \max D_s$ and $D' = \max D_s$. 

{\bf (Analysis Step 3):} By Lemma~\ref{lemma:farpoint} we have that whp, the number of nodes of $T$ at distance $\leq D'/5$ from $w$ and the number of nodes of $S$ at distance $\leq 2D'/5$ from $w$ are both $\leq \sqrt n$. Also if neither $D_1$, nor $D_2$ are good approximations, it must be that $d(s^*,X)>D/5$ and $d(s^*,Z)>2D/5$ and hence $D' \geq D$.

{\bf (Algorithm Step 4):}
Run BFS from $w$. Take all nodes of $S$ at distance $\leq 2D'/5$ from $w$, call these $S_w$, and run BFS from them. Whp, $|S_w|\leq \sqrt n$, so that this BFS run takes $O(m\sqrt n)$ time. Let $D_3:=\max_{s\in S_w,t\in T} d(s,t)$.

For every $s\in S_w$, let $t(s)$ be the closest node of $T$ to $s$ (breaking ties arbitrarily). Run BFS from each $t(s)$. Let $D_4:=\max_{s\in S_w,s'\in S} d(s',t(s))$.

{\bf (Analysis Step 4):} If $D_3\geq 3D/5$ or $D_4\geq 3D/5$, we are done, so let us assume that $D_3,D_4<3D/5$. Since $D_3<3D/5$, and since $D_3\geq d(w,t^*)$, it must be that $d(w,t^*)<3D/5$. Let $P_{wt^*}$ be the shortest $w$ to $t^*$ path. Consider the node $b$ on $P_{wt^*}$ for which $d(w,b)=2D/5$. If $b\in S$, then since $D'\geq D$, $b\in S_w$ and hence we ran BFS from $t(b)$. But since $d(b,t^*)=d(w,t^*)-2D/5<D/5$, and $d(b,t(b))\leq d(b,t^*)$ we have that $d(t(b),t^*)\leq 2D/5$ and hence $D_4\geq d(s^*,t(b))\geq D-d(t(b),t^*)\geq 3D/5$. Thus, if $D_4<3D/5$, it must be that $b\in T$.

%d(w,t^*)<=3d+e-y-1. b at dist 2d+y, 2D'/5>=2d+y. d(b,t^*)<=3d+e-y-1-2d-y = d+e-2y-1. d(t(b),t^*)<=2d+2e-4y-2., D_4\geq 5d+e-2d-2e+4y+2 = 3d - e+4y+2.

{\bf (Algorithm Step 5):} Take all nodes of $T$ at distance $\leq D'/5$ from $w$, call these $T_w$ and run BFS from them. Since $d(w,X)>D'/5$, whp $|T_w|\leq \sqrt n$, so this step runs in $O(m\sqrt n)$ time. Let $D_5=\max_{t\in T_w,s\in S} d(t,s)$.

{\bf (Analysis Step 5):} 
If $D_5\geq 3D/5$, we would be done, so assume that $D_5<3D/5$.
Let $a$ be the node on the shortest $w$ to $t^*$ path $P_{wt^*}$ with $d(w,a)=D/5$. Suppose that $a\in T$. Since $D'\geq D$, $a\in T_w$ and we ran BFS from it. However, also $d(a,t^*)=d(w,t^*)-d(w,a)<3D/5-D/5=2D/5$, and hence $D_5\geq d(a,s^*)\geq d(t^*,s^*)-d(t^*,a) \geq D-2D/5=3D/5$. Since $D_5<3D/5$, it must be that $a\in S$.
%a is d(w,a)=d+y/2., d(a,t^*)<= 3d+e-y-1-d-y/2 = 2d+e-3y/2-1. D_5\geq 5d+e-2d-e+3y/2+1 = 3d+3y/2+1.

Now, since $a\in S$ and $b\in T$, somewhere on the $a$ to $b$ shortest path $P_{ab}$, there must be an edge $(s',t')$ with $s'\in S,t'\in T$. Since $s'$ is before $b$, 
$d(w,s')\leq 2D/5\leq 2D'/5$, and hence $s'\in S_w$.
Thus we ran BFS from $t(s')$. 
Since $s'$ has an edge to $t'\in T$, $d(s',t(s'))\leq d(s',t')=1$. Also, since $d(w,s')\geq d(w,a) = D/5$ and $d(w,t^*)\leq 3D/5-1$, $d(s',t^*)\leq 2D/5-1$.
Thus, 
$$D_4\geq d(t(s'),s^*)\geq d(s^*,t^*)-d(t(s'),t^*) \geq D - d(t(s'),s')-d(s',t^*) \geq D-1-2D/5+1=3D/5.$$

%d(t(s'),t^*) \leq 1+d(s',t^*) \leq 1+(3d+e-y-1)-(d+y/2) = 2d+e-3y/2. so D_4\geq 3d+3y/2.
Hence if we set $D''=\max\{D_1,D_2,D_3,D_4,D_5\}$, we get that $3D/5\leq D''\leq D$.\end{proof}
%3d+e-y, 3d+3y/2, 3d+3y/2+1, 3d - e+4y+2
%e-y=3y/2 so e= 2.5y or 
%e-y=4y+2-e so 2e = 5y+2, and e=2.5y+1
%so if y=0.4e, we get at least 3d+0.6e
%but y can't be fractional and also y div by 2. 
%so e=1, y=0, we still get only 3d, so ceil(3D/5)=3d + 1, we get ceil(3D/5)-1.
%so e=2, y=0, we still get only 3d, so ceil(3D/5)=3d + 2, we get ceil(3D/5)-2.
%so e=3, y=2, we still get only 3d+1, so ceil(3D/5)=3d + 2, we get ceil(3D/5)-1.
%so e=4, y=2, we still get only 3d+2, so ceil(3D/5)=3d + 3, we get ceil(3D/5)-1.

%so e=1, y=0, we still get only 3d, so 3D/5=3d + 3/5, we get 3D/5-3/5.
%so e=2, y=0, we still get only 3d, so 3D/5=3d + 6/5, we get 3D/5-6/5.
%so e=3, y=2, we still get only 3d+1, so 3D/5=3d + 9/5, we get ceil(3D/5)-4/5.
%so e=4, y=2, we still get only 3d+2, so 3D/5=3d + 12/5, we get ceil(3D/5)-2/5.
%\end{proof}

We now modify the algorithm for unweighted graphs, both making the algorithm work for weighted graphs and removing the additive error, at the expense of increasing the runtime to $\tilde{O}(m^{3/2})$.

\begin{theorem}\label{thm:bichromedge}
There is an $\tilde{O}(m^{3/2})$ time algorithm, that given an undirected graph $G=(V,E)$ with nonnegative integer edge weights and $S\subseteq V, T=V\setminus S$, can output an estimate $D'$ such that $3D_{ST}(G)/5\leq D'\leq D_{ST}$.
\end{theorem}

\begin{proof}

Suppose as before the (bichromatic) $ST$-Diameter endpoints are $s^*\in S$ and $t^*\in T$ and that the $ST$-Diameter is $D$. 

{\bf (Algorithm Modified Step 1):}
 The algorithm here samples $E'\subseteq E$ of size $c \sqrt{m}\ln n$ uniformly at random, for large enough $c$. Let $Z$ be the endpoints of edges in $E'$ that are in $S$.
 %, and let $X$ be the endpoints that are in $T$. 
For every $z\in Z$, run Dijkstra's algorithm, and let $D_1 = \max_{z\in Z,t\in T} d(z,t)$.

{\bf (Analysis Step 1):} If for some $s'\in Z$ we have that $d(s^*,s')\leq 2D/5$, then $D_1\geq d(s',t^*)\geq D -  d(s^*,s') \geq 3D/5$. Let us then assume that $d(s^*,Z)>2D/5$.

%This will happen with high probability if the ball of $S$-nodes at distance $\leq 2D/5$ around $s^*$ has at least $\sqrt m$ edges incident to it. Otherwise, we can assume that there are fewer than $\sqrt m$ edges of the form $(s,b)$ where $d(s^*,s)\leq 2D/5$.
%if $d(s^*,s')\leq 2d+y$, $D_1\geq 3d+e-y$, so fewer ... dist $\leq 2d+y$.

{\bf (Algorithm Modified Step 2):} Let $X$ be the endpoints of edges in $E'$ that are in $T$. 
For every $t \in X$, run Dijkstra's algorithm and find the closest node $s(t)$ of $S$ to $t$.
Run Dijkstra's algorithm from every $s(t)$. Let $D_2=\max_{t\in X, t'\in T} d(s(t),t')$.

{\bf (Analysis Step 2):} If $s^*$ is at distance $\leq D/5$ from some node $t$ of $X$, then $d(s^*,s(t))\leq 2D/5$ (since $s(t)$ is closer to $t$ than $s^*$), and so $D_2\geq d(s(t),t^*)\geq 3D/5$. 
Let us then assume that $d(s^*,X)>D/5$.

As before, if we consider the nodes $M$ of $S$ that are at distance $> 2D/5$ from $Z$, then the node $w\in M$ that is furthest from $X$ among all nodes of $M$, would have both $d(w,Z)>2D/5$ and $d(w,X)>D/5$, as $s^*$ is in $M$ and satisfies $d(s^*,X)>D/5$. We will find a node $w$ with these properties in the next step.

{\bf (Algorithm Unmodified Step 3):} Perform exactly the same Step 3 as before, finding the largest integer $D'$ such that there is some node $w\in S$ with $d(w,Z)>2D'/5$ and $d(w,X)>D'/5$.

{\bf (Analysis Step 3):} Let $w\in S$ be the node we found such that $d(w,X)>D'/5,d(w,Z)>2D'/5$. By Lemma~\ref{lemma:edgefarpoint} we have that whp, the number of edges $(s,g)$ where $s\in S, g \in V$ and $d(w,s)\leq 2D'/5$ and the number of edges $(t,g')$ where $t\in T, g' \in V$ and $d(w,t)\leq D'/5$ is at most $\sqrt m$.
Also, if $D_1,D_2<3D/5$, then $D'\geq D$, so that we also have that the number of edges $(s,b)$ where $s\in S$ and $d(w,s)\leq 2D/5$ and the number of edges $(t,b')$ where $t\in T$ and $d(w,t)\leq D/5$ is at most $\sqrt m$, whp.

{\bf (Algorithm Modified Step 4):}
Run  Dijkstra's algorithm from $w$. Take all edges incident to nodes of $S$ at dist $\leq 2D'/5$ from $w$. Call these edges $E_S$ and their endpoints $S_w$.
Run  Dijkstra's algorithm from both of their end points. 
Whp, $|E_S|\leq \sqrt m$ and so $|S_w|\leq 2\sqrt m$, so that this Dijkstra run takes $\tilde{O}(m^{3/2})$ time. Let $D_3:=\max_{t\in S_w\cap T,s\in S} d(s,t)$.

For every $s\in S_w\cap S$, determine a closest node $t(s)\in T$ to $s$, and run  Dijkstra's algorithm from $t(s)$ as well. This search also takes $O(m^{3/2})$ time. Let $D_4:=\max_{s\in S_w\cap S,s'\in S} d(s',t(s))$.

{\bf (Analysis Step 4):} If $d(w,t^*)\geq 3D/5$, or $D_3\geq 3D/5$ or $D_4\geq 3D/5$, we are done, so let us assume that $d(w,t^*),D_3,D_4<3D/5$.

Now consider the node $b$ on the shortest $w$ to $t^*$ path $P_{wt^*}$ for which $d(w,b)\leq 2D/5$, but such that the node $b'$ after it on $P_{wt^*}$ has $d(w,b')>2D/5$.

Suppose that $b\in S$. Then since $D'\geq D$, we have $d(w,b)\leq 2D'/5$ and hence $(b,b')\in E_S$. Let us consider $d(b',t^*)=d(w,t^*)-d(b',w)$. Since $d(w,t^*)<3D/5$ and $d(b',w)>2D/5$, $d(b',t^*)<D/5$.
If $b'\in T$, then since we ran Dijkstra's algorithm from $b'$, we got $D_3\geq D-D/5=4D/5$. If $b'\in S$, then we ran Dijkstra's algorithm from $t(b')$ and $d(t(b'),t^*)\leq d(t(b'),b')+d(b',t^*)\leq 2d(b',t^*)<2D/5$, and hence $D_4\geq d(t(b),s^*)\geq D-2D/5=3D/5$.
Thus if neither $d(w,t^*)$, $D_3,$ nor $D_4$ are good approximations, then $b\in T$.

{\bf (Algorithm Modified Step 5):} 
Take all edges incident to nodes of $T$ at dist $\leq D'/5$ from $w$. Call these edges $E_T$ and their endpoints that are in $T$, $T_w$.
Run Dijkstra's algorithm from all nodes in $T_w$.

Since $d(w,X)>D'/5$, whp $|T_w|\leq 2\sqrt m$, so this step runs in $O(m^{3/2})$ time. Let $D_5=\max_{t\in T_w,s\in S} d(t,s)$.

{\bf (Analysis Step 5):} 
If $D_5\geq 3D/5$, we would be done, so assume that $D_5<3D/5$.
Let $a$ be the node on $P_{wt^*}$ with $d(w,a)\leq D/5$ but so that the node $a'$ after $a$ on $P_{wt^*}$ has $d(w,a')>D/5$.
Suppose that $a'\in T$. Since $D'\geq D$, $(a,a')\in E_T$, $a'\in T_w$ and we ran Dijkstra's algorithm from $a'$. However, also $d(a',t^*)=d(w,t^*)-d(w,a')<3D/5-D/5=2D/5$, and hence $D_5\geq d(a,s^*)\geq d(t^*,s^*)-d(t^*,a') \geq D-2D/5=3D/5$. Since $D_5<3D/5$, it must be that $a'\in S$.

Now, since $a'\in S$ and $b\in T$, somewhere on the $a'$ to $b$ shortest path $P_{ab}$, there must be an edge $(s',t')$ with $s'\in S,t'\in T$. However, since $s'$ is before $b$, we have that $d(w,s')\leq d(w,b)\leq 2D/5\leq 2D'/5$. Thus, $(s', t')\in E_S$ and we ran  Dijkstra's algorithm from $t'$. However, $d(t',t^*)=d(w,t^*)-d(w,t')\leq d(w,t^*)-d(w,a')< 3D/5-D/5=2D/5$, and hence $D_3\geq d(t',s^*)\geq d(s^*,t^*)-d(t',t^*)>3D/5$.

Hence if we set $D''=\max\{d(w,t^*),D_1,D_2,D_3,D_4,D_5\}$, we get that $3D/5\leq D''\leq D$.
\end{proof}

% ~5/3 approx in m^{3/2} time for weighted graphs, VIRGI, done

% ~5/3 approx for symmetric radius by modifying the above algorithms VIRGI

% 
% There is an $\tilde{O}(m\sqrt n)$ time algorithm, that given an unweighted graph $G=(V,E)$ and $S\subseteq V, T=V\setminus S$, 
% that outputs an estimate 
% $R'$ such that either

% $\min\{R_{ST}(G),R_{TS}(G)\} \leq R'\leq 5/3 \cdot \max\{R_{ST}(G),R_{TS}(G)\}$
%  if $\tilde{R}_{ST}$ is divisible by $3$, and otherwise $\min\{R_{ST}(G),R_{TS}(G)\} \leq R'\leq 5/3 \cdot \max\{R_{ST}(G),R_{TS}(G)\}+2$.

\subsection{Undirected Bichromatic Radius}
We begin with a simple near-linear time algorithm that achieves almost a $2$-approximation.

\begin{theorem}Let $G=(V,E)$ be an undirected graph with nonnegative edge weights $w$. Let $S\subseteq V, T=V\setminus S$. There is an $O(m+n\log n)$ time algorithm that outputs an estimate $R'$ such that $R_{ST}\leq R'\leq 2R_{ST}+\min_{s\in S, t\in T, (s,t)\in E} w(s,t)$.
If $G$ is unweighted, the algorithm runs in $O(m+n)$ time and $R_{ST}\leq R'\leq 2R_{ST}+1$.
\end{theorem}

\begin{proof}
The algorithm is as follows. Let $(s,t)\in E$ be the smallest weight edge among those with $s\in S,t\in T$. Run Dijkstra's algorithm from $s$ and output $R'=\max_{t'\in T} d(s,t')$.

Clearly $R_{ST}\leq R'$. Let $s^*\in S$ be the true $ST$-center. Then for all $t'\in T$,  $d(s,t')\leq d(s,s^*)+R_{ST}$. On the other hand, $d(s,s^*)\leq w(s,t)+d(t,s^*)\leq w(s,t)+R_{ST}$, and hence $R'\leq w(s,t)+2R_{ST}$.

For unweighted graphs, $w(s,t)=1$ and we can run BFS instead of Dijkstra's algorithm.
\end{proof}

We now present a $\tilde{O}(m\sqrt{n})$ algorithm for Bichromatic Radius, similar in spirit to our Bichromatic Diameter algorithm.

\begin{theorem}\label{thm:bichromedgeradius}
There is an $\tilde{O}(m\sqrt{n})$ time algorithm, that given an undirected unweighted graph $G=(V,E)$ and $S\subseteq V, T=V\setminus S$, can output estimates $R'_{ST}$ such that 
$R_{ST}\leq R'_{ST}\leq 5R_{ST}/3 + 5/3$. If $R_{ST}$ is divisible by $3$, $R_{ST}\leq R'_{ST}\leq 5R_{ST}/3+1$.
\end{theorem}

\begin{proof}
Let $s^*\in S$ be the $ST$-center of $G$ and let $R=R_{ST}$ be the $ST$-Radius.

{\bf (Algorithm Step 1):}
The algorithm samples $S_1\subseteq S$ of size $c \sqrt{n} \ln n$ uniformly at random. For every $s\in S_1$, run BFS and find $t(s)\in T$ which is closest to $s$. Let $T_2=\{t(s)~|~s\in S_1\}$.

Then sample $T_1\subseteq T$ of size $c \sqrt{n} \ln n$ uniformly at random. For every $t\in T_1$, run BFS and find $s(t)\in S$ which is closest to $t$. Let $S_2=\{s(t)~|~t\in T_1\}$.

Let $s_0\in S$ be the node minimizing $\max_{t\in T_1\cup T_2} d(s_0,t)$. 
Let $R_1=\max_{t\in T} d(s_0,t)$.
Let $w\in T$ be the node maximizing $d(w,T_1\cup T_2)$.

{\bf (Analysis Step 1):} 
We know that $\max_{t\in T_1\cup T_2} d(s^*,t)\leq R$, and hence 
$\max_{t\in T_1\cup T_2} d(s_0,t)\leq R$.

Suppose that for every $t\in T$, $d(t,T_1\cup T_2)\leq 2R/3$. Then, $d(s_0,t)\leq R+2R/3=5R/3$ and hence $R_1\leq 5R/3$ and $s_0$ would be a good approximate center. Thus, we can assume that there exists some $t$ with $d(t,T_1\cup T_2)>2R/3$, and in particular, $d(w,T_1\cup T_2)>2R/3$.

Moreover, suppose that there is some $s\in S_1$ such that $d(w,s)\leq R/3$. Then, $d(w,t(s))\leq d(w,s)+d(s,t(s))\leq 2d(w,s)\leq 2R/3$, contradicting the fact that $d(w,T_1\cup T_2)>2R/3$. Thus, we must have that $d(w,S_1)>R/3$.

Now, since $T_1$ is random of size $c\sqrt n\ln n$, by Lemma~\ref{lemma:farpoint}, the number of nodes of $T$ at distance $\leq 2R/3$ from $w$ is at most $\sqrt n$, whp. Similarly, since $S_1$ is random of size $c\sqrt n\ln n$, by Lemma~\ref{lemma:farpoint}, the number of nodes of $S$ at distance $\leq R/3$ from $w$ is at most $\sqrt n$, whp.

{\bf (Algorithm Step 2):}
Run  BFS from $w$. Take the closest $\sqrt n$ nodes $T_w$ of $T$ at distance from $w$. 
Run BFS from all $t\in T_w$, and find $s(t)\in S$ closest to $t$. Run BFS from each $s(t)$. 

Let $R_2:=\min_{t'\in T_w} \max_{t\in T} d(s(t'),t)$.

{\bf (Analysis Step 2):} 
Since $|T_w|\leq \sqrt n$, the runtime of this step is $O(m\sqrt n)$.

Since $w\in T$, we know that $d(w,s^*)\leq R$.
Now consider the node $b$ on the shortest $w$ to $s^*$ path $P_{ws^*}$ for which $d(w,b)\leq 2R/3$, but such that the node $b'$ after it on $P_{wt^*}$ has $d(w,b')>2R/3$. Since the graph is unweighted, we get that $d(w,b)=\lfloor 2R/3\rfloor\geq 2R/3-2/3.$

%When $R$ is divisible by $3$ as we assumed, $d(w,b)=2R/3$. When $R$ is not divisible by $R$ this is one source of additive error.

Let us consider $d(b,s^*)=d(w,s^*)-d(w,b)$. Since $d(w,s^*)\leq R$ and $d(w,b)\geq 2R/3-2/3$, $d(b,s^*)\leq R/3+2/3$.

Suppose that $b\in T$. By our previous argument, as $d(w,b)\leq 2R/3$, $b$ must be in $T_w$. Then we ran BFS from $s(b)$ and $d(s(b),s^*)\leq d(s(b),b)+d(b,s^*)\leq 2d(b,s^*)\leq 2R/3+4/3$, and hence $R_2\leq 2R/3+R+4/3=5R/3+4/3$.
Thus if $R_2$ is not a good approximation, then $b\in S$.

{\bf (Algorithm Step 3):} 
Take the $\sqrt n$ closest nodes of $S$ to $w$. Call these $S_w$. Run BFS from every $s\in S_w$.
Set $R_3:=\min_{s\in S_w} \max_{t\in T} d(s,t)$.

{\bf (Analysis Step 3):} 
Since $|S_w|\leq \sqrt n$, the runtime of this step is $O(m\sqrt n)$.

Let $a$ be the node on $P_{ws^*}$ with $d(w,a)\leq R/3$ but so that the node $a'$ after $a$ on $P_{ws^*}$ has $d(w,a')>R/3$. We have that $d(w,a)=\lfloor R/3\rfloor\geq R/3-2/3.$

Suppose that $a\in S$. As $d(w,a)\leq R/3$ and $a$ is among the closest $\sqrt n$ nodes to $w$ by our previous argument, we ran BFS from $a$.

However, also $d(a,s^*)=d(w,s^*)-d(w,a)\leq R-R/3+2/3=2R/3+2/3$, and hence $R_3\leq 2R/3+R+2/3=5R/3+2/3$. If $R_3$ is not a good approximation, it must be that $a\in T$.

Now, since $a\in T$ and $b\in S$, somewhere on the $a$ to $b$ shortest path $P_{ab}$, there must be an edge $(t',s')$ with $s'\in S,t'\in T$.
However, since $t'$ is before $b$, we have that $d(w,t')\leq d(w,b)\leq 2R/3$. Thus, $t'\in T_w$ and we ran BFS from $s(t')$. However,
$d(t',s(t'))\leq d(t',s')=1$, and hence $d(s(t'),s^*)\leq d(s(t'),t')+d(t',s^*)\leq 1+d(w,s^*)-d(w,t')\leq 1+R-d(w,a)=2R/3+5/3$.
Hence for every $t\in T$, $d(s(t'),t)\leq 5R/3+5/3$.
If $R$ is divisible by 3, the only source of additive error is the $+1$ from using the edge $(t',s(t'))$ instead of $(t',s')$.

Hence if we set $R'=\min \{R_1,R_2,R_3\}$, we have
$R\leq R'\leq 5R/3+5/3$. 
If $R$ is divisible by 3, $R\leq R'\leq 5R/3+1$.
\end{proof}

We now use edge sampling to remove the additive error and make the algorithm work for weighted graphs as well, at the expense of increasing the runtime to $\tilde{O}(m^{3/2})$.

\begin{theorem}\label{thm:birad}
There is an $\tilde{O}(m^{3/2})$ time algorithm, that given an undirected graph $G=(V,E)$ with nonnegative integer edge weights and $S\subseteq V, T=V\setminus S$, can output estimates $R'_{ST}$ such that $R_{ST}\leq R'_{ST}\leq 5R_{ST}/3$.
\end{theorem}

\begin{proof}
Let $s^*\in S$ be the $ST$-center of $G$ and let $R=R_{ST}$ be the $ST$-Radius.

{\bf (Algorithm Step 1):}
We sample $c\sqrt m\ln n$ edges $E'\subseteq E$ uniformly at random.
Let $S_1$ be the endpoints that are in $S$ and let $T_1$ be the endpoints in $T$.
For every $s\in S_1$, run Dijkstra and find $t(s)\in T$ which is closest to $s$. Let $T_2=\{t(s)~|~s\in S_1\}$.

For every $t\in T_1$, run Dijkstra and find $s(t)\in S$ which is closest to $t$. Let $S_2=\{s(t)~|~t\in T_1\}$.

Let $s_0\in S$ be the node minimizing $\max_{t\in T_1\cup T_2} d(s_0,t)$. Run Dijkstra from $s_0$.
Let $R_1=\max_{t\in T} d(s_0,t)$.
Let $w\in T$ be the node maximizing $d(w,T_1\cup T_2)$.

{\bf (Analysis Step 1):} 
The algorithm runs in $\tilde{O}(m^{3/2})$ time.

We know that $\max_{t\in T_1\cup T_2} d(s^*,t)\leq R$, and hence 
$\max_{t\in T_1\cup T_2} d(s_0,t)\leq R$.

Suppose that for every $t\in T$, $d(t,T_1\cup T_2)\leq 2R/3$. Then, $d(s_0,t)\leq R+2R/3=5R/3$ and hence $R_1\leq 5R/3$ and $s_0$ would be a good approximate center. Thus, we can assume that there exists some $t$ with $d(t,T_1\cup T_2)>2R/3$, and in particular, $d(w,T_1\cup T_2)>2R/3$.

Moreover, suppose that there is some $s\in S_1$ such that $d(w,s)\leq R/3$. Then, $d(w,t(s))\leq d(w,s)+d(s,t(s))\leq 2d(w,s)\leq 2R/3$, contradicting the fact that $d(w,T_1\cup T_2)>2R/3$. Thus, we must have that $d(w,S_1)>R/3$.

Now, since $E'$ is random of size $c\sqrt m\ln n$, by
% an edge version of Lemma~\ref{lemma:farpoint}
Lemma~\ref{lemma:edgefarpoint}, the number of edges $(t,g)$ where
$t\in T, g \in V$ and $d(w,t)\leq 2R/3$ is at most $\sqrt m$, whp. Similarly, the number of edges $(s,g)$ where $s\in S, g \in V$ and $d(s,w)\leq R/3$ is at most $\sqrt m$, whp.

{\bf (Algorithm Step 2):}
Run Dijkstra from $w$. Consider the edges $(t,b)$ with $t\in T$ sorted in nondecreasing order according to $d(w,t)$. 
Let $E^T$ be the first $\sqrt m$ edges in this sorted order. Run Dijkstra from both endpoints of each edge in $E^T$. Call $T_w$ those endpoints that are in $T$ and $S^1_w$ those in $S$.
Let $R_2:=\min_{s\in S^1_w} \max_{t\in T} d(s,t)$.

For every $t\in T_w$, determine a closest node $s(t)\in T$ to $t$, and run  Dijkstra's algorithm from $s(t)$ as well. Let $R_3:=\min_{t\in T_w}\max_{t'\in T} d(s(t),t')$.

{\bf (Analysis Step 2):} 
Since $|E^T|\leq \sqrt m$, the runtime of this step is $\tilde{O}(m^{3/2})$.

If $R_2\leq 5R/3$ or $R_3\leq 5R/3$, we are done. So let us assume that $R_2,R_3>5R/3$. Also, since $w\in T$, we know that $d(w,s^*)\leq R$.

Now consider the node $b$ on the shortest $w$ to $s^*$ path $P_{ws^*}$ for which $d(w,b)\leq 2R/3$, but such that the node $b'$ after it on $P_{ws^*}$ has $d(w,b')>2R/3$.

Suppose that $b\in T$. 
Then since $d(w,b)\leq 2R/3$ and since by the previous argument the edges from $T$ nodes at distance $2R/3$ from $w$ is at most $\sqrt m$, $(b,b')$ must be among the edges in $E^T$. We thus run Dijkstra's from both $b$ and $b'$.

Let us consider $d(b',s^*)=d(w,s^*)-d(w,b')$. Since $d(w,s^*)\leq R$ and $d(w,b')>2R/3$, $d(b',s^*)<R/3$.
If $b'\in S$, then since we ran Dijkstra's algorithm from $b'$, we got $R_2\leq 4R/3$. If $b'\in T$, then we ran Dijkstra's algorithm from $s(b')$ and $d(s(b'),s^*)\leq d(s(b'),b')+d(b',s^*)\leq 2d(b',s^*)<2R/3$, and hence $R_3\leq 2R/3+R=5R/3$.
Thus if neither $R_2,$ nor $R_3$ are good approximations, then $b\in S$.

{\bf (Algorithm Step 3):} 
Consider the edges $(s,b)$ with $s\in S$ sorted in nondecreasing order according to $d(w,s)$. 
Let $E^S$ be the first $\sqrt m$ edges in this sorted order. Run Dijkstra from both endpoints of each edge in $E^S$. Call $S^2_w$ those endpoints that are in $S$.
Let $R_4:=\min_{s\in S^2_w} \max_{t\in T} d(s,t)$.

{\bf (Analysis Step 3):} As $|E^S|=\sqrt m$,  $|S_w|\leq 2\sqrt m$, so this step runs in $\tilde{O}(m^{3/2})$ time.

If $R_4\leq 5R/3$, we would be done, so assume that $R_4>5R/3$.
Let $a$ be the node on $P_{ws^*}$ with $d(w,a)\leq R/3$ but so that the node $a'$ after $a$ on $P_{ws^*}$ has $d(w,a')>R/3$.
Suppose that $a'\in S$. Then since $d(w,a)\leq R/3$, $(a,a')\in E^S$, $a'\in S^2_w$ and we ran Dijkstra's algorithm from $a'$. However, also $d(a',s^*)=d(w,s^*)-d(w,a')<R-R/3=2R/3$, and hence $R_4\leq 2R/3+R=5R/3$. Since $R_4>5R/3$, it must be that $a'\in T$.

Now, since $a'\in T$ and $b\in S$, somewhere on the $a'$ to $b$ shortest path $P_{ab}$, there must be an edge $(t',s')$ with $s'\in S,t'\in T$. However, since $t'$ is before $b$, we have that $d(w,t')\leq d(w,b)\leq 2R/3$. Thus, $(t', s')\in E^T$ and we ran  Dijkstra's algorithm from $s'$. However, $d(s',s^*)=d(w,s^*)-d(w,s')\leq d(w,s^*)-d(w,a')< R-R/3=2R/3$, and hence $R_2\leq R+2R/3=5R/3$.

Hence if we set $R'=\min \{R_1,R_2,R_3,R_4\}$, we have
$R\leq R'\leq 5R/3$ 
\end{proof}

\subsection{Undirected Bichromatic Eccentricities.}
In the next section we will give approximation algorithms for $ST$-Eccentricities in undirected graphs which imply algorithms for bichromatic Eccentricities in undirected graphs with same guarantees. We reproduce them here for convenience. 

\begin{proposition}\label{biecc3approx}
There is an $O(m+n\log{n})$ time algorithm, that given an undirected graph $G=(V,E)$ with nonnegative integer edge weights and $S\subseteq V,T = V \setminus S$, can output an estimate $\eps_{ST}'(v)$ for each node $v\in S$ such that $\eps_{ST}(v)/3\le \eps_{ST}'(v) \le\eps_{ST}(v)$.
\end{proposition}

\begin{theorem}\label{biecc2approx}
There is an $\tilde{O}(m\sqrt n)$ time algorithm, that given an unweighted graph $G=(V,E)$ and $S\subseteq V,T = V \setminus S$, can output an estimate $\eps_{ST}'(v)$ for each $v\in S$ such that $\eps_{ST}(v)/2-5/2\le \eps_{ST}'(v)\le \eps_{ST}(v)$. If $\eps_{ST}(v)$ is divisible by $2$, $\eps_{ST}(v)/2-2\le \eps_{ST}'(v)\le \eps_{ST}(v)$.
\end{theorem}

\begin{theorem}\label{biecc2approx-edge}
There is an $\tilde{O}(m^{3/2})$ time algorithm, that given an undirected graph $G=(V,E)$ with nonnegative integer edge weights and $S\subseteq V,T = V \setminus S$, can output estimates $\eps_{ST}'(v)$ for each $v\in S$, such that $\eps_{ST}(v)/2\le \eps_{ST}'(v)\leq \eps_{ST}(v)$.
\end{theorem}

\subsection{Directed Bichromatic Diameter}

\begin{theorem}\label{thm:bidirdiam}
There is an $\tilde{O}(m^{3/2})$ time algorithm, that given a directed graph $G=(V,E)$ with nonnegative integer weights and $S\subseteq V, T=V\setminus S$, can output an estimate $D'$ such that $D_{ST}(G)/2\leq D'\leq D_{ST}(G)$. 
\end{theorem}

\begin{proof}
Suppose the (bichromatic) $ST$-Diameter endpoints are $s^*\in S$ and $t^*\in T$ and that the $ST$-Diameter is $D$. The algorithm does not know $D$, but we will use it in the analysis. 

{\bf (Algorithm Step 1):}
The algorithm first samples $E' \subseteq E$ of size $c \sqrt{m} \ln m$ for large enough $c$ uniformly at random from the edges which go from $S$ to $T$. Let $R$ be the set of $S$ nodes incident to these edges. Define $D_1 = \max_{u\in R,t\in T} d(u,t)$.

{\bf (Analysis Step 1):} If there exists an $s \in R$ with $d(s^*, s) \leq D/2$ then we are done as by triangle inequality $D_1 \geq d(s, t^*) \geq d(s^*, t^*)-d(s^*, s) \geq D/2$.

{\bf (Algorithm Step 2):} Let $w$ be the vertex in $S$ which maximizes $d(w, R)$. Defining the distance to an edge $(u, v)$ to be distance to $u$ we find the $\sqrt{m}$ closest edges to $w$ which cross from $S$ to $T$. Let $P$ be the set of $T$ nodes incident to these edges. Let $D_2 = \max_{s \in S, v\in P} d(s, v)$ and $D_3 = \max_{t \in T} d(w, t)$. Our estimate is $D' = \max(D_1, D_2, D_3)$.

{\bf (Analysis Step 2):} Note that all $3$ estimates are underestimates so we will just bound $D'$ from below. Suppose $D_3 \geq D/2$ then we are already done. So we can assume that $d(w, t^*) < D/2$. Let $(s, t)$ be the first edge going from $S$ to $T$ in the shortest path from $w$ to $t^*$. If $D_1 < D/2$ then by 
% an edge version of Lemma~\ref{lemma:farpoint} 
Lemma~\ref{lemma:edgefarpoint}, this edge is among the $\sqrt{m}$ closest edges to $w$. Hence $D_2 \geq d(s^*, t) \geq d(s^*, t^*)-d(t, t^*) \geq D-d(t, t^*) \geq D-d(w, t^*) \geq D/2$
\end{proof}

\section{Algorithms for $ST$-Eccentricities and Radius}\label{sec:stalgs}
All of the algorithms in this section are for undirected graphs; we later prove that the directed versions of these problems do not admit truly subquadratic algorithms with any finite approximation factor.

We do not give algorithms for $ST$-Diameter, as tight algorithms were already given in~\cite{diamstoc18}.

\subsection{$ST$-Eccentricities} We begin with a near-linear time $3$-approximation algorithm. 

\begin{proposition}\label{STecc3approx}
There is an $O(m+n\log{n})$ time algorithm, that given an undirected graph $G=(V,E)$ with nonnegative integer edge weights and $S,T\subseteq V$, can output an estimate $\eps_{ST}'(v)$ for each node $v\in S$ such that $\eps_{ST}(v)/3\le \eps_{ST}'(v) \le\eps_{ST}(v)$.
\end{proposition}
% If eps(v) is divisible by 2 or by 4?

\begin{proof}
The algorithm is as follows. Pick an arbitrary node $t\in T$ and run Dijkstra's algorithm from it. Let $t'$ be a node in $T$ maximizing $d(t',t)$, and run Dijkstra's algorithm from $t'$. For each $v\in S$, output $\eps_{ST}'(v)=\max\{d(v,t),d(v,t')\}$.

Clearly $\eps_{ST}'(v)\le \eps_{ST}(v)$. Now suppose that $v'\in T$ is the farthest node from $v$ in $T$. So we have $\eps_{ST}(v) = d(v,v') \le d(v,t)+d(t,v') \le d(v,t)+d(t,t') \le d(v,t)+d(t,v)+d(v,t') \le 3\eps_{ST}'(v)$, where the first and third inequalities are from triangle inequality and the second inequality is from the definition of $t'$. 
\end{proof}

Now we turn to our $2$-approximation algorithms.
Our first theorem is for unweighted graphs. Later on, we modify the algorithm in this theorem to obtain an algorithm for weighted graphs as well, and at the same time remove the small additive error that appears in the theorem below.

\begin{theorem}\label{STecc2approx}
There is an $\tilde{O}(m\sqrt n)$ time algorithm, that given an undirected unweighted graph $G=(V,E)$ and $S,T\subseteq V$, can output an estimate $\eps_{ST}'(v)$ for each $v\in S$ such that $\eps_{ST}(v)/2-5/2\le \eps_{ST}'(v)\le \eps_{ST}(v)$. If $\eps_{ST}(v)$ is divisible by $2$, $\eps_{ST}(v)/2-2\le \eps_{ST}'(v)\le \eps_{ST}(v)$.
\end{theorem}

\begin{proof}
For each $v\in S$, let $v'$ be the farthest node from $v$, i.e. $d(v,v')=\eps_{ST}(v)$. 

{\bf (Algorithm Step 1):} The algorithm samples $X\subset V$ of size $c\sqrt{n}\ln{n}$ uniformly at random. For every $x\in X$, run BFS and find $t(x)\in T$ which is closest to $x$ (if $x\in T$, $t(x)=x$). Let $T_X=\{t(x)|x\in X\}$.

Run BFS from each node $t\in T_X$. For each $v\in S$ let $\eps^{(1)}_{ST}(v) = \max_{t\in T_X}{d(v,t)}$.

Let $w\in T$ be the node maximizing $d(w,T_X)$. 

{\bf (Analysis Step 1):} This step of the algorithm runs in $\tilde{O}(m\sqrt{n})$.

Suppose there is some node $t\in T_X$ such that $d(v',t) \le \eps_{ST}(v)/2$. Then $d(v,t) \ge d(v,v')-d(v',t)\ge \eps_{ST}(v)/2$, and so $\eps^{(1)}_{ST}(v) $ is a good approximation for $v$. Thus we can assume that $d(v',T_X)>\eps_{ST}(v)/2$, and so $d(w,T_X)>\eps_{ST}(v)/2$. Now since $X$ is random of size $c\sqrt{n}\ln{n}$, by Lemma \ref{lemma:farpoint}, the number of nodes of $T$ at distance $\le \eps_{ST}(v)/2$ from $w$ is at most $\sqrt{n}$ whp.

Moreover, suppose that there is some node $x\in X$ such that $d(w,x)\le \eps_{ST}(v)/4$. Then $d(w,t(x))\le d(w,x)+d(x,t(x))\le2d(w,x)\le\eps_{ST}(v)/2$, contradicting the fact that $d(w,T_X)>\eps_{ST}(v)/2$. Thus, we must have that $d(w,X)>\eps_{ST}(v)/4$. 

Now, since $X$ is random of size $c\sqrt{n}\ln{n}$, by Lemma \ref{lemma:farpoint}, the number of nodes at distance $\le \eps_{ST}(v)/4$ from $w$ is at most $\sqrt{n}$ whp.  

{\bf (Algorithm Step 2):} Run BFS from $w$. For each $v\in S$, let $\eps^{(2)}_{ST}(v) = d(v,w)$.

Take the closest $\sqrt{n}$ nodes of $V\setminus T$ to $w$. Call these $Y$. Run BFS from all $y\in Y$, and let $e(y)=\max_{t\in T} d(y,t)$. Let $\eps^{(3)}_{ST}(v) = \max_{y\in Y}e(y)-d(v,y)$.

{\bf (Analysis Step 2):} If $d(v,w)\ge \eps_{ST}(v)/2$, then $\eps^{(2)}_{ST}(v)$ is a good estimate. So assume that $d(v,w)\le \lceil \eps_{ST}/2 \rceil - 1\le \eps_{ST}/2-1/2$.

Now consider the node $a$ on the shortest $w$ to $v$ path $P_{wv}$ for which $d(w,a)\le \eps_{ST}(v)/4$, but such that the node $a'$ after it on $P_{wv}$ has $d(w,a')>\eps_{ST}(v)/4$. Since the graph is unweighted, we get that $d(w,a)=\lfloor \eps_{ST}(v)/4 \rfloor\ge \eps_{ST}(v)/4-3/4$. 

If $a\in V\setminus T$, then by the previous argument since $d(a,w)\le \eps_{ST}(v)/4$, $a\in Y$ and we run BFS from $a$. Since $e(a)\ge d(a,v')\ge d(v,v')-d(a,v)$ and $d(a,v)=d(w,v)-d(a,w)$, we have $e(a) \ge 3\eps_{ST}(v)/4-1/4$. So $e(a)-d(v,a) \ge \eps_{ST}(v)/2-1/2$. Moreover, if $a'$ is the farthest node from $a$ in $T$, then $\eps_{ST}(v)\ge d(v,a')\ge d(a,a')-d(v,a)=e(a)-d(v,a)$, and hence $\eps^{(3)}_{ST}(v)$ is a good estimate. 

So assume that $a\in T$.

{\bf (Algorithm Step 3):} Take the closest $\sqrt{n}$ nodes of $T$ to $w$. Call these $T_w$. Run BFS from all $t\in T_w$ and find $y(t)\in V\setminus T$. Run BFS from each $y(t)$, and let $e(y(t)) = \max_{t'\in T} d(y(t),t')$. Let $\eps^{(4)}_{ST}(v) = \max_{t\in T_w}e(y(t))-d(v,y(t))$.

{\bf (Analysis Step 3):}
Consider the node $b$ on $P_{wv}$ for which $d(w,b)\le 3\eps_{ST}(v)/8$, but such that the node $b'$ after it on $P_{wv}$ has $d(w,b')>3\eps_{ST}(v)/8$. Since the graph is unweighted, we get that $d(w,b)=\lfloor 3\eps_{ST}(v)/8 \rfloor\ge 3\eps_{ST}(v)/8-7/8$. 

If $b\in T$, then since $d(w,b)\le\eps_{ST}(v)/2$, by previous argument $b\in T_w$ and we run BFS from $b$. Since $d(v,b)=d(w,v)-d(w,b) \le \eps_{ST}(v)/8 +3/8$, we have that $d(v,y(b))\le d(v,b)+d(b,y(b)) \le 2d(v,b)\le \eps_{ST}(v)/4 + 3/4$. Similar to the previous step, we get that $e(y(b))-d(v,y(b)) \ge d(y(b),v')-d(v,y(b)) \ge d(v,v')-2d(v,y(b)) \ge \eps_{ST}(v)/2-3/2$. By considering the farthest node from $y(b)$ in $T$, we can show that $e(y(b))-d(v,y(b))\le \eps_{ST}(v)$ and hence $\eps^{(4)}_{ST}(v)$ is a good approximate. So if $\eps^{(4)}_{ST}(v)$ is not a good approximate, it must be that $b\in V\setminus T$.

Now, since $a\in T$ and $b\in V\setminus T$, somewhere on the $a$ to $b$ shortest path $P_{ab}$, there must be an edge $(t',y')$ with $t'\in T$ and $y'\in V\setminus T$. However, since $t'$ is on $P_{wv}$, we have $d(w,t')\le d(v,w)<\eps_{ST}(v)/2$. Thus, $t'\in T_w$ and we run BFS from $y(t')$. However, $d(t', y(t'))\le d(t',y') =1$, and hence $d(y(t'),v)\le d(y(t'),t')+d(t',v)\le 1+ d(w,v)-d(w,t') \le 1+ \eps_{ST}(v)/2-1/2-d(w,a) \le \eps_{ST}(v)/4 + 5/4$. So we get that 
\begin{equation*}
    e(y(t'))-d(y(t'),v)\ge d(y(t'),v')-d(y(t'),v)\ge d(v,v')-2d(y(t'),v) \ge \eps_{ST}(v)/2-5/2
\end{equation*}

Moreover, if $y'$ is the farthest node $y(t')$ in $T$, then $\eps_{ST}(v)\ge d(v,y')\ge d(y',y(t'))-d(v,y(t'))=e(y(t'))-d(v,y(t'))$.
Hence if for each $v\in S$ we set $\eps_{ST}'(v) = \max\{\eps^{(1)}_{ST}(v),\eps^{(2)}_{ST}(v),\eps^{(3)}_{ST}(v),\eps^{(4)}_{ST}(v)\}$, we have $\eps_{ST}(v)/2-5/2\le \eps_{ST}'(v)\le \eps_{ST}(v)$.
\end{proof}

We now use edge sampling to remove the additive error from the above algorithm and make the algorithm work for weighted graphs as well, at the expense of increasing the runtime to $\tilde{O}(m^{3/2})$.

\begin{theorem}\label{STecc2approx-edge}
There is an $\tilde{O}(m^{3/2})$ time algorithm, that given an undirected graph $G=(V,E)$ with nonnegative integer edge weights and $S,T\subseteq V$, can output estimates $\eps_{ST}'(v)$ for each $v\in S$, such that $\eps_{ST}(v)/2\le \eps_{ST}'(v)\leq \eps_{ST}(v)$.
\end{theorem}

\begin{proof}
For each $v\in S$, let $v'$ be the farthest node from $v$, i.e. $d(v,v')=\eps_{ST}(v)$.

{\bf (Algorithm Step 1):} We sample $c\sqrt{m}\ln{n}$ edges $E'\subseteq E$ uniformly at random. Run Dijkstra from both endpoints of edges in $E'$ (we call these vertices $V(E')$), and for each endpoint $x$, find $t(x)\in T$ which is closest to $x$. Let $T_{E'} = \{t(x)| x\in V(E')\}$.

Run Dijkstra from each node in $T_{E'}$, and for each $v\in S$, let $\eps_{ST}^{(1)}(v) = d(v,T_{E'})$. 

Let $w\in T$ be the node maximizing $d(w,T_{E'})$.

{\bf (Analysis Step 1):} Since $V(E')=\tilde{O}(\sqrt{m})=|T_{E'}|$, this step takes $\tilde{O}(m^{3/2})$ time.

Suppose there is some node $t\in T_{E'}$ such that $d(v',t)\le \eps_{ST}(v)/2$. Then $d(v,t) \ge d(v,v')-d(v',t)\ge \eps_{ST}(v)/2$, and so $\eps^{(1)}_{ST}(v) $ is a good approximation for $\eps_{ST}(v)$. Thus we can assume that $d(v',T_{E'})>\eps_{ST}(v)/2$, and so $d(w,T_{E'})>\eps_{ST}(v)/2$. Now since $E'$ is random of size $c\sqrt{m}\ln{n}$, by 
%an edge version of Lemma \ref{lemma:farpoint}
Lemma~\ref{lemma:edgefarpoint}, the number of edges $(t,g)$ where $t\in T, g \in V$ and $d(w,t)\le \eps_{ST}(v)/2$ is at most $\sqrt{m}$, whp.

Moreover, suppose that there is some edge $(x,b)\in E'$ such that $d(w,x)\le \eps_{ST}(v)/4$. Then $d(w,t(x))\le d(w,x)+d(x,t(x))\le2d(w,x)\le\eps_{ST}(v)/2$, contradicting the fact that $d(w,T_{E'})>\eps_{ST}(v)/2$. Thus, we must have that $d(w,V(E'))>\eps_{ST}(v)/4$. 

Now, since $E'$ is random of size $c\sqrt{n}\ln{n}$, by Lemma \ref{lemma:edgefarpoint}, the number of edges $(x,g)\in E'$ where $g \in V$ such that $d(w,x)\le\eps_{ST}(v)/4 $ is at most $\sqrt{m}$, whp. 

{\bf (Algorithm Step 2):} Run Dijkstra from $w$. For each $v\in S$, let $\eps^{(2)}_{ST}(v) = d(v,w)$.

Consider the edges $(y,b)$ sorted in nondecreasing order according to $d(w,y)$. Let $E''$ be the first $\sqrt{m}$ edges in this sorted order. Let $Y$ be the endpoints of edges in $E''$ that are in $V\setminus T$. Run Dijkstra from each node in $Y$ and let $e(y)=\max_{t\in T} d(y,t)$. Let $\eps^{(3)}_{ST}(v) = \max_{y\in Y}e(y)-d(v,y)$.

{\bf (Analysis Step 2):} Since $|Y|=\tilde{O}(\sqrt{m})$, this step takes $\tilde{O}(m^{3/2})$ time.

If $d(v,w)\ge \eps_{ST}(v)/2$, then $\eps^{(2)}_{ST}(v)$ is a good approximation. So assume that $d(v,w)< \eps_{ST}(v)/2$.

Now consider the node $a$ on the shortest $w$ to $v$ path $P_{wv}$ for which $d(w,a)\le \eps_{ST}(v)/4$, but such that the node $a'$ after it on $P_{wv}$ has $d(w,a')>\eps_{ST}(v)/4$. 

Since $d(w,a)\le \eps_{ST}(v)/4$, by the previous argument the number of edges from the nodes at distance $\eps_{ST}(v)/4$ from $w$ is at most $\sqrt{m}$, and so $(a,a')$ must be among the edges in $E''$. Suppose that $a'\in V\setminus T$. We thus run Dijkstra from $a'$.

Let us consider $d(a',v) = d(w,v)-d(w,a')$. Since $d(w,a')>\eps_{ST}(v)/4$ and $d(w,v)<\eps_{ST}(v)/2$, $d(a',v)<\eps_{ST}(v)/4$. Thus $e(a') \ge d(a',v') \ge d(v,v')-d(a',v) > 3\eps_{ST}(v)/4$. So $e(a')-d(a',v)>\eps_{ST}(v)/2$. Now if $a''$ is the farthest node from $a'$ in $T$, then $\eps_{ST}(v) \ge d(v,a'') \ge d(a',a'')-d(v,a') = e(a')-d(v,a')$, and hence $\eps^{(3)}_{ST}(v)$ is a good approximation. 

So we assume that $a'\in T$.

{\bf (Algorithm Step 3):} Consider the edges $(t,b)$ with $t\in T$ sorted in nondecreasing order according to $d(w,t)$. Let $E^T$ be the first $\sqrt{m}$ edges in this sorted order. Run Dijkstra from both endpoints of each edge in $E^T$ (call these nodes $V(E^T)$), and find $y(x)\in V\setminus T$ closest to $x$, for each $x\in V(E^T)$. Run Dijkstra from each $y(x)$, and let $e(y(x)) = \max_{t\in T} d(y(x),t)$. Let $\eps^{(4)}_{ST}(v) = \max_{x\in V(E^T)}e(y(x))-d(v,y(x))$.

{\bf (Analysis Step 3):}

Consider the node $b$ on $P_{wv}$ for which $d(w,b)\le 3\eps_{ST}(v)/8$, but such that the node $b'$ after it on $P_{wv}$ has $d(w,b')>3\eps_{ST}(v)/8$. 

Suppose that $b'\in T$, then since $d(w,b)\le\eps_{ST}(v)/2$, by the previous argument $(b,b')\in E^T$ and we run Dijkstra from $y(b')$. Let us consider $d(y(b'),v)\le d(v,b') +d(b',y(b'))\le 2d(v,b')$. Since $d(v,b') = d(v,w)-d(w,b') < \eps_{ST}(v)/8$, $d(y(b'),v)< \eps_{ST}(v)/4$. Similar as in the previous step, we get that $\eps_{ST}(v)\ge e(y(b'))-d(y(b'),v)$ and also $e(y(b'))-d(y(b'),v) \ge d(y(b'),v')-d(y(b'),v) \ge d(v,v')-2d(y(b'),v) > \eps_{ST}(v)/2$, thus $\eps^{(4)}_{ST}(v)$ is a good approximation. So if $\eps^{(4)}_{ST}(v)$ is not a good approximation, it must be that $b'\in V\setminus T$.

Now, since $a'\in T$ and $b'\in V\setminus T$, somewhere on the $a'$ to $b'$ shortest path $P_{a'b'}$, there must be an edge $(t,x)$ with $t\in T$ and $x\in V\setminus T$. However, since $t$ is on $P_{wv}$, we have $d(w,t)\le d(v,w)<\eps_{ST}(v)/2$. Thus, $(t,x)\in E^T$ and we run Dijkstra from $x$. 

Let us consider $y(x)$. Since $x\in V\setminus T$, $y(x)=x$. Moreover since $x$ is after $a'$ on $P_{a'b'}$, $d(w,x)\ge d(w,a') > \eps_{ST}(v)/4$, and thus $d(x,v)=d(v,w)-d(w,x) < \eps_{ST}(v)/4$. So $e(y(x))-d(y(x),v)=e(x)-d(x,v) \ge d(x,v')-d(x,v) \ge d(v,v')-2d(x,v) > \eps_{ST}(v)/2$. 

Hence if for each $v\in S$ we set $\eps_{ST}'(v) = \max\{\eps^{(1)}_{ST}(v),\eps^{(2)}_{ST}(v),\eps^{(3)}_{ST}(v),\eps^{(4)}_{ST}(v)\}$, we have $\eps_{ST}(v)/2\le \eps_{ST}'(v)\le \eps_{ST}(v)$.
\end{proof}

\subsection{$ST$-Radius}\label{sec:radiustrick} %In this part we provide upper bounds on ST-radius obtained by taking the minimum over the ST-eccentricity of all the vertices in $S$.

A simple argument shows that given any approximation algorithm for $ST$-Eccentricities, one obtains an approximation algorithm for $ST$-Radius with the same approximation factor. First, run the $ST$-Eccentricities algorithm and let $v$ be the vertex with the smallest estimated Eccentricity $\epsilon'(v)$. Then run Dijkstra's algorithm from $v$ and report $\epsilon_{ST}(v)$ as the $ST$-Radius estimate $R'$. Let $R$ be the true $ST$-Radius of the graph and let $c$ be the true $ST$-center. If $\alpha$ is the approximation ratio for the $ST$-Eccentricities algorithm then $\epsilon_{ST}(v)\leq \alpha\epsilon'(v)\leq \alpha\epsilon_{ST}(v)$ and $\epsilon_{ST}(c)\leq \alpha\epsilon'(c)\leq \alpha\epsilon_{ST}(c)$. By choice of $v$, $\epsilon'(v)\leq \epsilon'(c)$. Thus, $\alpha R=\alpha\epsilon_{ST}(c)\geq \alpha\epsilon'(c)\geq \alpha\epsilon'(v)\geq \epsilon_{ST}(v)=R'$. Clearly $R'\geq R$, so $R\leq R'\leq \alpha R$.

Thus, we get the following theorems from our algorithms for $ST$-Eccentricities.
\begin{theorem}
There is an $O(m+n\log{n})$ time algorithm, that given an undirected graph $G=(V,E)$ with nonnegative integer edge weights and $S,T\subseteq V$, can output an estimate $R'$ such that $R_{ST}/3\le R' \le R_{ST}$.
\end{theorem}

% \begin{proof}
% The algorithm is as follows. Run the ST-eccentricity algorithm of Theorem \ref{STecc2approx}. Let $R'=\min_{v\in S} \eps_{ST}'(v)$. 

% We prove that $R'$ is a good approximation. Let $c$ be the ST-center, i.e. $\eps_{ST}(c) = \min_{v\in S} \eps_{ST}(v)$. Let $c'= \arg{\min_{v\in S} \eps_{ST}'(v)}$. By Theorem \ref{STecc3approx}, we know that $R'=\eps_{ST}'(c')\ge \eps_{ST}(c')/3 \ge R_{ST}/3$. Moreover, $R'\le \eps_{ST}'(c)\le \eps_{ST}(c)=R_{ST}$.
% \end{proof}

% The following theorems are similarly derived from Theorem \ref{STecc2approx} and Theorem \ref{STecc2approx-edge}, respectively.

\begin{theorem}
There is an $\tilde{O}(m\sqrt n)$ time algorithm, that given an undirected unweighted graph $G=(V,E)$ and $S,T\subseteq V$, can output an estimate $R'$ such that $R_{ST}/2-5/2\le R'\leq R_{ST}$.
\end{theorem}

\begin{theorem}\label{STradius2approx-edge}
There is an $\tilde{O}(m^{3/2})$ time algorithm, that given an undirected graph $G=(V,E)$ with nonnegative integer edge weights and $S,T\subseteq V$, can output estimates $R'$ such that $R_{ST}/2\le R'\leq R_{ST}$.
\end{theorem}

%{\bf ST-symmetric radius.} The algorithms for ST-radius clearly can be used for the symmetric radius problem, by running it twice, on each of the two sets $S$ and $T$ as the $"S"$ set and taking the maximum over the resulting values. 

\section{Algorithms for Subset Diameter, Eccentricities, and Radius}\label{sec:subalgs}
We obtain 2-approximations for Subset Diameter in directed graphs and Subset Radius in undirected graphs simply by running Dijkstra's algorithm from an arbitrary vertex $s\in S$. We obtain an almost 2-approximation in almost linear time for directed Subset Eccentricities (and thus directed Subset Radius) by a slight modification of an algorithm for (non-Subset) Eccentricities in directed graphs from~\cite{diamstoc18}. 
\begin{proposition}[Directed Subset Diameter]\label{prop:subdiam}
There is an $\tilde{O}(m)$ time algorithm, that given a directed graph $G=(V,E)$ with nonnegative integer weights and $S\subseteq V$, outputs an estimate $D'$ such that $D_S/2\leq D'\leq D_S$.
\end{proposition}

\begin{proof}$ $
Run Dijkstra's algorithm both ``forward" and ``backward" from $s$ to obtain $D_1=\max_{s'\in S}d(s,s')$ and $D_2=\max_{s'\in S}d(s',s)$. Return $D'=\max\{D_1,D_2\}$. 

Let $s^*,t^*\in S$ be the true endpoints of the Subset Diameter. Then, by the triangle inequality $D_S\leq d(s^*,s)+d(s,t^*)$. Then since $d(s^*,s)\leq D_2$ and $d(s,t^*)\leq D_1$, $D_S\leq D_1+D_2$. Thus, $D_S/2\leq \max\{D_1,D_2\}\leq D_S$.
\end{proof}
\begin{proposition}[Undirected Subset Radius]\label{prop:subrad}
There is an $\tilde{O}(m)$ time algorithm, that given an undirected graph $G=(V,E)$ with nonnegative integer weights and $S\subseteq V$, outputs an estimate $R'$ such that $R_S/2\leq D'\leq R_S$.
\end{proposition}

\begin{proof}
Run Dijkstra's algorithm from $s$ and return $R'=\max_{s'\in S}d(s,s')$.

Let $c^*\in S$ be the true center. Then since $d(c^*,s')\leq R_S$ for all $s'\in S$, the triangle inequality implies that for all $s'$, $d(s,s')\leq 2R_S$. Thus, $R_S\leq R'\leq 2R_S$.
\end{proof}

%note: this alg is copy pasted from other paper with only a few characters of modification
\begin{theorem}[Directed Subset Eccentricities]\label{thm:subecc}
Suppose that we are given a directed graph $G=(V,E)$ with nonnegative integer weights. For any $1>\tau>0$ we can in $\tO(m/\tau)$ time output for all $v\in S$ an estimate $\eps'(v)$ such that $\frac{1-\tau}{2}\eps_S(v)\leq \eps'(v)\leq \eps_S(v)$.
\end{theorem}

\begin{proof}

The algorithm proceeds in iterations and maintains a set $U$ of nodes for which we still do not have a good Eccentricity estimate. In each iteration either we get a good estimate for many new vertices and hence remove them from $U$, or we remove all vertices from $U$ that have large Eccentricities, and for the remaining nodes in $U$ we have a better upper bound on their Eccentricities. After a small number of iterations we have a good estimate for all vertices of the graph.  Initially $U=S$ and we will end with $|U|\leq O(1)$. When $|U|\leq O(1)$ we can evaluate $\eps_S(v)$ for all $v \in U$ in the total time of $O(m)$. 

Also we maintain a value $D$ that upper bounds the largest Eccentricity of a vertex in $U$. That is, $\eps_S(v)\leq D$ for all $v \in U$. Initially we set $D=n^C$ for some large enough constant $C>0$ (we assume that the set $S$ is strongly connected). The algorithm proceeds in phases. Each phase takes $\tO(m)$ time and either $|U|$ decreases by a factor of at least $2$ or $D$ decreases by a factor of at least $1/(1-\tau)$. After $O(\log(n)/\tau)$ phases either $|U|\leq O(1)$ or $D<1$. 

	For a subset $U \subseteq V$ of vertices and a vertex $x \in V$ we define a set $U_x\subseteq S$ to contain those $|U_x|=|U|/2$ vertices from $U$ that are closest to $x$ (according to distance $d(\cdot,x)$). The ties are broken by taking the vertex with the smaller id.
	Given a subset $U \subseteq V$ of vertices and a threshold $D$, a phase proceeds as follows.
	\begin{itemize}[itemsep=0mm]
		\item We sample a set $A \subseteq U$ of $O(\log n)$ random vertices from the set $U$. By Lemma~\ref{lemma:disthit}, with high probability for all $x \in V$ we have $A \cap U_x \neq \emptyset$.
		\item Let $w$ be the vertex in $S$ that maximizes $d(A,w)$. We can find it by constructing a vertex $y$ adjacent to every vertex in $A$ and running Dijkstra's algorithm from $y$.
		\item We consider two cases.
		%	\begin{itemize}
\paragraph*{Case $d(U\setminus U_w,w)\geq {1-\tau \over 2}D$.} For all $x \in U\setminus U_w$ we have ${1-\tau \over 2} D \leq \eps_S(x)\leq D$ and we assign the estimate $\eps'(x)={1-\tau \over 2}D$. This gives us that ${1-\tau \over 2} \eps_S(x)\leq \frac{1-\tau}{2}D=\eps'(x)\leq \eps_S(x)$ for all $x\in U\setminus U_w$. We update $U$ to be $U_w$. This decreases the size of $U$ by a factor of $2$ as required.
\paragraph*{Case $d(U\setminus U_w,w)<{1-\tau \over 2}D$.} Set $U'=U$. For every vertex $v \in U$ evaluate $r_v:=\max_{x \in A}d(v,x)$. We can evaluate these quantities by running Dijkstra's algorithm from every vertex in $A$ and following the incoming edges. If $r_v\geq {1-\tau \over 2}D$, then assign the estimate $\eps'(v)={1-\tau \over 2}D$ and remove $v$ from $U'$. Similarly as in the previous case we have ${1-\tau \over 2} \eps_S(v)\leq\eps'(v)\leq \eps_S(v)$ for all $v \in U\setminus U'$. Below we will show that for every $v \in U'$ we have $\eps_S(v)\leq (1-\tau)D$. Thus we can update $U=U'$ and decrease the threshold $D$ to $(1-\tau)D$ as required.
		%	\end{itemize}
	\end{itemize}
	
	\paragraph*{Correctness} We have to show that, if there exists $v \in U'$ such that $\eps_S(v)>(1-\tau)D$, then we will end up in the first case (this is the contrapositive of the claim in the second case). Since $v \in U'$ we must have that $d(v,x)\leq {1-\tau \over 2}D$ for all $x \in A$. Since $\eps_S(v)>(1-\tau)D$, we must have that there exists $v'\in S$ such that $d(v,v')>(1-\tau)D$. By the triangle inequality we get that $d(x,v')>{1 - \tau \over 2}D$ for every $x \in A$. By choice of $w$, we have $d(A,w)>{1-\tau \over 2}D$. Since $A \cap U_w \neq \emptyset$, we have $d(U\setminus U_w,w)\geq {1-\tau \over 2}D$ and we will end up in the first case.
	
	The guarantee on the approximation factor follows from the description.
\end{proof}

\paragraph*{Directed Subset Radius}
Using the argument from Section~\ref{sec:radiustrick}, we obtain an algorithm for Directed Subset Radius from our algorithm for Directed Subset Eccentricities.

\begin{theorem}[Directed Subset Radius]\label{thm:subrad}
Suppose that we are given a directed graph $G=(V,E)$ with nonnegative integer weights. For any $1>\tau>0$ we can in $\tO(m/\tau)$ time output an estimate $R'$ such that $R_{S}\leq R'\leq \frac{2}{1-\tau}R_S$.
\end{theorem}

\section{Parameterized Algorithms for Bichromatic Diameter, Radius, and Eccentricities}\label{sec:paramalgs}

In this section we give algorithms for Bichromatic Diameter, Radius, and Eccentricities with runtimes parameterized by the size of the \emph{boundary} $B$. Let $S'$ be the set of vertices in $S$ that have a neighbor in $T$ and let $T'$ be the set of vertices in $T$ that have a neighbor in $S$. Let $B$ be whichever of $S'$ or $T'$ is smaller in size. 

\subsection{Undirected Parameterized Bichromatic Diameter}
\begin{theorem}\label{thm:paramdiam}
There is an $O(m|B|)$ time algorithm, that given an unweighted undirected graph $G=(V,E)$ and $S\subseteq V, T=V\setminus S$, outputs an estimate $D'$ such that $2D_{ST}(G)/3-1\leq D'\leq D_{ST}(G)$.
\end{theorem}

\begin{proof}
{\bf (Algorithm):} For all $v\in T$, we let $\eps_{ST}(v)=\max_{s\in S}d(s,v)$ ($\eps_{ST}(v)$ is already defined for $v\in S$). Suppose without loss of generality that $B\subseteq S$ (a symmetric argument works for $B\subseteq T$). For every vertex $v\in B$, run BFS from $v$, let $v_T$ be an arbitrary neighbor of $v$ such that $v_T\in T$, and run BFS from $v_T$. Let $D_1$ be the largest $S-T$ distance found. That is, $D_1=\max_{v\in B}\max\{\eps_{ST}(v),\eps_{ST}(v_T)\}$. Let $s\in S$ be the farthest vertex from $B$. That is, $s$ is the vertex in $S$ that maximizes $d(s, B)$. Then, we run BFS from $s$ and let $D_2=\eps_{ST}(s)$. Return $D'=\max\{D_1,D_2\}$.

{\bf (Analysis):} Let $s^*\in S,t^*\in T$ be the true endpoints of the Bichromatic Diameter and let $D$ denote $D_{ST}(G)$. If $s^*$ is of distance at most $D/3+1$ from some vertex $v\in B$ then by the triangle inequality $d(v,t^*)\geq 2D/3-1$ so $D_1\geq 2D/3-1$ and we are done. If $t^*$ is of distance at most $D/3$ from some vertex $v\in B$ then by the triangle inequality $d(v_T,s^*)\geq 2D/3-1$ so  $D_1\geq 2D/3-1$ and we are done. 

Now, if we are not already done, $s^*$ is of distance at least $D/3+1$ from every vertex in $B$, so $s$ is also of distance at least $D/3+1$ from every vertex in $B$. Additionally, $t^*$ is of distance at least $D/3$ from every vertex in $B$. We observe that the shortest path between $s$ and $t^*$ must contain a vertex in $B$. Thus, $d(s,t^*)=\min_{v\in B}d(s,v)+d(v,t^*)\leq (D/3+1)+(D/3)=2D/3+1$. Thus, $D_2\geq 2D/3+1$ and we are done.
\end{proof}
\subsection{Undirected Parameterized Bichromatic Radius}
\begin{theorem}\label{thm:paramrad}
There is an $O(m|B|)$ time algorithm that, given an unweighted undirected graph $G=(V,E)$ and $S\subseteq V, T=V\setminus S$, returns an estimate $R'$ such that $R_{ST}{G}\leq R'\leq 3R_{ST}(G)/2+3$.
\end{theorem}

\begin{proof}
{\bf (Algorithm):} If $B\subseteq S$, we run BFS from all $v\in B$ and let $R_1$ be the minimum Eccentricity found; that is, $R_1=\min_{v\in B}\eps_{ST}(v)$. If $B\subseteq T$, for every $v\in B$, we let $v_S$ be an arbitrary neighbor of $v$ such that $v_S\in S$, and run BFS from $v_S$. In this case we let $R_1=\min_{v\in B}\eps_{ST}(v_S)$. Let $U$ be the set of vertices that we have run BFS from so far.

Then, let $s\in S$ be the vertex that is closest to $\emph{all}$ vertices in $U$; that is, let $s$ be the vertex that minimizes $\max_{v\in U}d(s,v)$. Run BFS from $s$ and let $R_2=\eps_{ST}(s)$. Return $\min\{R_1,R_2\}$.

{\bf (Analysis):} Let $c^*\in S$ be the true center and let $R$ denote $R_{ST}(G)$; that is, $\eps_{ST}(c^*)=R$. If there exists a vertex $v\in U$ such that $d(c^*,v)\leq R/2$, then since $U\subseteq S$ and by the triangle inequality, $\eps_{ST}(v)\leq 3R/2$ and we are done. 

If we are not done by the previous step, $c^*$ must be of distance at least $R/2$ from \emph{every} vertex in $U$, and thus of distance at least $R/2-1$ from every vertex in $B$. We observe that the shortest path between $s$ and any vertex in $T$ must contain a vertex in $B$. Thus, every vertex in $T$ must be of distance at most $R/2+1$ from \emph{some} vertex in $B$, and thus of distance at most $R/2+2$ from some vertex in $U$.

Since for all $v\in T$, $d(c^*,v)\leq R$, the triangle inequality implies that for all $v\in U$, $d(c^*,v)\leq R+1$. Therefore, by choice of $s$, for all $v\in U$, $d(s,v)\leq R+1$. We claim that $\eps_{ST}(s)\leq 3R/2$. Consider an arbitrary vertex $t\in T$. Let $u$ be a vertex in $U$ such that $d(u,t)\leq R/2+2$; such a $u$ exists by the previous paragraph. Then, $d(s,u)+d(u,t)\leq (R+1)+(R/2+2)=3R/2+3$. Thus, $\eps_{ST}(s)\leq 3R/2+3$.
\end{proof}

\subsection{Undirected Parameterized Bichromatic Eccentricities}
\begin{theorem}\label{thm:paramecc}
There is an $O(m|B|)$ time algorithm that, given an unweighted undirected graph $G=(V,E)$ and $S\subseteq V, T=V\setminus S$, returns for every $v\in S$ an estimate $\eps'(v)$ such that $3\eps_{ST}(v)/5-1\leq \eps'(v)\leq \eps_{ST}(v)$.
\end{theorem}

\begin{proof}
{\bf (Algorithm):} Suppose $B\subseteq S$. For every vertex $u\in B$, we run BFS from $u$, let $u'$ be the vertex in $T$ that maximizes $d(u,u')$, and run BFS from $u'$. Then for every vertex $u\in B$ we let $u_T$ be an arbitrary neighbor of $u$ such that $u_T\in T$ and run BFS from $u_T$. Then, let $t\in T$ be the farthest vertex from $B$; that is, $t$ is the vertex in $T$ that maximizes $d(B, t)$. Let $T''$ be the set of vertices in $T$ that we have run BFS from. For every vertex $v\in S$, we return the estimate $\eps'(v)=\max_{t''\in T''}d(v,t'')$.

We use a similar algorithm for when $B\subseteq T$: For every vertex $u\in B$, we run BFS from $u$, let $u'$ be the vertex in $T$ that maximizes $d(u,u')$, and run BFS from $u'$. Then, let $t\in T$ be the farthest vertex from $B$; that is, $t$ is the vertex in $T$ that maximizes $\min_{u\in B}d(u,t)$. Let $T''$ be the set of vertices in $T$ that we we have run BFS from. For every vertex $v\in S$, we return the estimate $\eps'(v)=\max_{t''\in T''}d(v,t'')$.

{\bf (Analysis):} Suppose $B\subseteq S$. If there exists a vertex in $u\in B$ such that $d(v,u)\geq 3\eps_{ST}(v)/5$, then $d(v,u_T)\geq 3\eps_{ST}(v)/5-1$ so we are done. On the other hand, suppose $B\subseteq T$. If there exists a vertex in $u\in B$ such that $d(v,u_T)\geq 3\eps_{ST}(v)/5-1$, then we are done. Otherwise, $v$ is of distance at most $3\eps_{ST}(v)/5$ from \emph{every} vertex in $B$. Thus, regardless of whether $B\subseteq S$ or $T$, if we are not already done, $v$ is of distance at most $3\eps_{ST}(v)/5$ from \emph{every} vertex in $B$.

Then, since every path from $v$ to any vertex in $T$ must contain a vertex in $B$, there must exist a vertex in $T$ that is of distance at least $2\eps_{ST}(v)/5$ from \emph{every} vertex in $B$. In particular, $t$ must be of distance at least $2\eps_{ST}(v)/5$ from every vertex in $B$.

Let $v'$ be the true farthest vertex from $v$; that is, $d(v,v')=\eps_{ST}(v)$. If there exists a vertex in $u\in B$ such that $d(v,u)\leq \eps_{ST}(v)/5$, then by the triangle inequality $d(u,v')\geq 4\eps_{ST}(v)/5$, so $d(u,u')\geq 4\eps_{ST}(v)/5$. Applying the triangle inequality again, $d(v,u')\geq 3\eps_{ST}(v)/5$, so we are done. Otherwise, every vertex $u\in B$ is of distance at least $\eps_{ST}(v)/5$ from $v$.

We claim that if we are not already done, $d(v,t)\geq 3\eps_{ST}(v)/5$. We observe that every path from $v$ to $t$ must contain a vertex in $B$. Let $u\in B$ be a vertex on the shortest path from $v$ to $t$. Then, $d(v,t)=d(v,u)+d(u,t)\geq \eps_{ST}(v)/5+2\eps_{ST}(v)/5=3\eps_{ST}(v)/5$.
\end{proof}

% For directed graphs, we refine the definition of $B$. Let $S'$ be the set of vertices in $S$ with an outgoing edge to a vertex in $T$. Let $T'$ be the set of vertices in $T$ with an incoming edge from a vertex in $S$. Let $B$ be whichever of $S'$ or $T'$ is smaller in size. 

% \begin{corollary}
% There is an $O(m|B|)$ time algorithm that returns an estimate $D'$ for the Bichromatic Diameter $D$ of a directed unweighted graph such that $D/2\leq D'\leq D$. 
% \end{corollary}

% Instead of sampling from the boundary, we just run BFS from every vertex in the boundary. 

%Recall that $S'$ is the set of vertices in $S$ that have a neighbor in $T$, $T'$ is the set of vertices in $T$ that have a neighbor in $S$, and $B$ is whichever of $S'$ or $T'$ is smaller in size. 
\subsection{Directed Parameterized Bichromatic Diameter}
For Bichromatic Diameter in undirected graphs, we assumed that only one of $S'$ or $T'$ was small (i.e. we set $B$ to be the smaller of the two); however for directed graphs we impose that both $S'$ and $T'$ are small, by defining a new parameter $B'=S'\cup T'$. 

\begin{theorem}\label{thm:paramdirdiam}
There is an $O(m|B'|)$ time algorithm that, given an unweighted directed graph $G=(V,E)$ and $S\subseteq V, T=V\setminus S$, returns an estimate $D'$ such that $2D_{ST}(G)/3\leq D'\leq D_{ST}(G)$. 
\end{theorem}

\begin{proof}
{\bf (Algorithm):} For all $v\in T$, we let $\eps_{ST}(v)$ denote $\max_{s\in S}d(s,v)$ ($\eps_{ST}(v)$ is already defined for $v\in S$). Run forward BFS from every vertex in $S'$ and run backward BFS from every vertex in $T'$. Let $D_1$ be the largest $S\rightarrow T$ distance found. That is, $D_1=\max_{v\in B'}\eps_{ST}(v)$. Let $s\in S$ be the farthest vertex from $B'$. That is, $s$ is the vertex in $S$ that maximizes $d(s, B')$. Then, we run BFS from $s$ and let $D_2=\eps_{ST}(s)$. Return $\max\{D_1,D_2\}$.

{\bf (Analysis):} Let $s^*\in S$ and $t^*\in T$ be the true Bichromatic Diameter endpoints and let $D$ denote $D_{ST}(G)$. If there exists a vertex $s'\in S'$ such that $d(s^*,s')\leq D/3$, then by the triangle inequality, $d(s',t^*)\geq 2D/3$ so $D_1\geq 2D/3$ and we are done. Similarly, if there exists a vertex $t'\in T'$ such that $d(t',t^*)\leq D/3$, then by the triangle inequality, $d(s^*,t')\geq 2D/3$ so $D_1\geq 2D/3$ and we are done. 

Suppose we are not done. Then, for every vertex $s'\in S'$, $d(s^*,s')>D/3$ and for every vertex $t'\in T'$, $d(t',t^*)>D/3$. By choice of $s$, for all $s'\in S'$, $d(s,s')>D/3$. We observe that every path from $s$ to $t^*$ must contain an edge from a vertex in $S'$ to a vertex in $T'$. Let $(s''\in S',t''\in T')$ be an edge on the shortest path from $s$ to $t^*$. Then, $d(s,t^*)=d(s,s'')+d(s'',t'')+d(t'',t^*)>D/3+1+D/3=2D/3+1$, so $D_2\geq 2D/3+1$.
\end{proof}

%todos:
%parameterized bichrom radius should really have an additive factor of only 1 by separating B\subseteq S and T cases

\section{Conditional Lower Bounds}\label{sec:lbs}

\subsection{Bichromatic Diameter, Eccentricities, and Radius}
\paragraph{Undirected Bichromatic Diameter}
The following theorem implies that our algorithms for undirected Bichromatic Diameter from Theorem~\ref{thm:bichromedge} and Proposition~\ref{prop:linbidiam} are tight under SETH. 

\begin{theorem}\label{thm:bidiamlb}
Under SETH, for every $k\geq 2$, every algorithm that can distinguish between Bichromatic Diameter $2k-1$ and $4k-3$ in undirected unweighted graphs requires $m^{1+1/(k-1)-o(1)}$ time. 
\end{theorem}
In particular setting $k=2$ and 3 in Theorem~\ref{thm:bidiamlb} implies that our $m^{3/2}$ time $5/3$-approximation algorithm from Theorem~\ref{thm:bichromedge} is tight in approximation factor and runtime, respectively. Furthermore, setting $k$ to be arbitrarily large implies that our $\tilde{O}(m)$ time almost 2-approximation algorithm from Proposition~\ref{prop:linbidiam} is tight under SETH.
% \begin{corollary}
% Under SETH, any algorithm for Bichromatic Diameter that achieves a $5/3-\delta$ approximation factor for $\delta>0$ in $m$-edge undirected unweighted graphs requires $m^{2-o(1)}$ time.

% Under SETH, any algorithm for Bichromatic Diameter that achieves a $9/5-\delta$ approximation factor for $\delta>0$ in $m$-edge undirected unweighted graphs requires $m^{3/2-o(1)}$ time.
% \end{corollary}

Theorem~\ref{thm:bidiamlb} follows from the following lemma.
\begin{lemma}
	Let $k\geq 2$ be any integer. Given a $k$-OV instance, we can in $O(kn^{k-1}d^{k-1})$ time construct an unweighted, undirected graph with $O(k n^{k-1}+k n^{k-2} d^{k-1})$ vertices and $O(kn^{k-1}d^{k-1})$ edges that satisfies the following two properties.
	\begin{enumerate}
		\item If the $k$-OV instance has no solution, then for \emph{all} pairs of vertices $u\in S$ and $v\in T$ we have $d(u,v)\leq 2k-1$.
		\item If the $k$-OV instance has a solution, then there exists a pair of vertices $u\in S$ and $v\in T$ such that $d(u,v)\geq 4k-3$.
	\end{enumerate}
\end{lemma}

\begin{proof} $ $
%\paragraph*{
\\
{\bf Construction of the graph.} We begin with the $k$-OV-graph from Theorem~\ref{thm:kOV}. Additionally, we add $k-1$ new layers of vertices $L_{k+1},\dots,L_{2k-1}$, where each new layer contains $n^{k-1}$ vertices and is connected to the previous layer by a matching. That is, each new layer contains one vertex for every tuple $(a_1,\dots,a_{k-1})$ where $a_i\in W_i$ for all $i$, and each $(a_1,\dots,a_{k-1})\in L_j$ is connected to its counterpart $(a_1,\dots,a_{k-1})\in L_{j-1}$ by an edge, for all $j$.

We let $S=L_0$ and we let $T$ contain the rest of the vertices in the graph.

%include figure?
%
%\paragraph*
\noindent{\bf Correctness of the construction.} 

\emph{Case 1: The $k$-OV instance has no solution.} By property 3 of Theorem~\ref{thm:kOV} for all $u\in S$ and $v\in L_k$, $d(u,v)=k$. Then, since $L_{k},\dots,L_{2k-1}$ form a series of matchings, for all $u\in S$ and $v\in L_{k+1}\cup\dots\cup L_{2k-1}$, $d(u,v)\leq 2k-1$. Furthermore, property 5 of Theorem~\ref{thm:kOV} implies that for all $u\in S$ and $v\in L_1\cup\dots\cup L_{k-1}$, $d(u,v)\leq 2k-1$. Thus, we have shown that for all $u\in S$ and $v\in T$ we have $d(u,v)\leq 2k-1$.

\noindent\emph{Case 2: The $k$-OV instance has a solution.} Let $(a_0,a_1,\ldots, a_{k-1})$  be a solution to the $k$-OV instance where $a_i\in W_i$ for all $i$. We claim that $d((a_0,\dots,a_{k-2}) \in S,(a_1,\dots,a_{k-1}) \in L_{2k-1}))\geq 4k-3$. Since $L_{k},\dots,L_{2k-1}$ form a series of matchings, every path from $(a_0,\dots,a_{k-2}) \in S$ to $(a_1,\dots,a_{k-1})\in L_{2k-1}$ contains the vertex $(a_1,\dots,a_{k-1})\in L_{k}$. By property 4 of Theorem~\ref{thm:kOV}, $d((a_0,\dots,a_{k-2}) \in S,(a_1,\dots,a_{k-1}) \in L_{k})\geq 3k-2$. Thus, $d((a_0,\dots,a_{k-2}) \in S,(a_1,\dots,a_{k-1}) \in L_{2k-1}))\geq 4k-3$. 
\end{proof}

% ~5/3 approx in m^{3/2} time for weighted graphs, VIRGI, done

% ~5/3 approx for symmetric radius by modifying the above algorithms VIRGI

% 
% There is an $\tilde{O}(m\sqrt n)$ time algorithm, that given an unweighted graph $G=(V,E)$ and $S\subseteq V, T=V\setminus S$, 
% that outputs an estimate 
% $R'$ such that either

% $\min\{R_{ST}(G),R_{TS}(G)\} \leq R'\leq 5/3 \cdot \max\{R_{ST}(G),R_{TS}(G)\}$
%  if $\tilde{R}_{ST}$ is divisible by $3$, and otherwise $\min\{R_{ST}(G),R_{TS}(G)\} \leq R'\leq 5/3 \cdot \max\{R_{ST}(G),R_{TS}(G)\}+2$.

\paragraph*{Undirected Bichromatic Eccentricities}

The following proposition implies that our algorithms for undirected Bichromatic Eccentricities from Theorem~\ref{biecc2approx-edge} and Proposition~\ref{biecc3approx} are tight under SETH. 
\begin{proposition}\label{prop:whoknows}
Under SETH, for every $k\geq 2$, every algorithm that can distinguish between Bichromatic Eccentricities $k$ and $3k-2$ in undirected unweighted graphs requires $m^{1+1/(k-1)-o(1)}$ time. 
\end{proposition}
In particular setting $k=2$ and 3 in Theorem~\ref{prop:whoknows} implies that our $m^{3/2}$ time $2$-approximation algorithm from Theorem~\ref{biecc2approx-edge} is tight under SETH in approximation factor and runtime, respectively. Furthermore, setting $k$ to be arbitrarily large implies that our $\tilde{O}(m)$ time almost 3-approximation algorithm from Proposition~\ref{biecc3approx} is tight under SETH.

Proposition~\ref{prop:whoknows} follows from the following lemma.
\begin{lemma} 
	Let $k\geq 2$ be any integer. Given a $k$-OV instance, we can in $O(kn^{k-1}d^{k-1})$ time construct an unweighted, undirected graph with $O(k n^{k-1}+k n^{k-2} d^{k-1})$ vertices and $O(kn^{k-1}d^{k-1})$ edges that satisfies the following two properties. Let $S_0$ be a particular subset of $S$.
	\begin{enumerate}
		\item If the $k$-OV instance has no solution, then for \emph{all} vertices $v\in S_0$ we have $\eps_{ST}(v)\leq k$.
		\item If the $k$-OV instance has a solution, then there exists a vertex $v\in S_0$ such that $\eps_{ST}(v)\geq 3k-2$.
	\end{enumerate}
\end{lemma}

\begin{proof}
We begin with the $k$-OV-graph from Theorem~\ref{thm:kOV}. Let $T=L_k$ and let $S$ contain the rest of the vertices in the graph. Let $S_0=L_0$.

If the $k$-OV instance has no solution then by property 3 of Theorem~\ref{thm:kOV} for all $u\in L_0$ and $v\in T$, $d(u,v)=k$. Thus, for all $u\in L_0$, $\eps_{ST}(u)=k$. 

Suppose the $k$-OV instance has a solution $(a_0,\dots,a_{k-1})$. Then by property 4 of Theorem~\ref{thm:kOV}, $d((a_0,\dots,a_{k-2}) \in L_0,(a_1,\dots,a_{k-1}) \in T)\geq 3k-2$, so $\eps_{ST}(a_0,\dots,a_{k-2})\geq 3k-2$.
\end{proof}

\paragraph*{Undirected Bichromatic Radius}

%5/3-eps approx needs n^{2-o(1)} time under HS VIRGi
%5/3-eps approx needs n^{2-o(1)} time under HS. VIRGI

%no finite approx. possible for directed in subquadratic time under HS,  VIRGI
The following theorem implies that our $\tilde{O}(m^{3/2})$ time $5/3$-approximation algorithm for undirected Bichromatic Radius from Theorem~\ref{thm:birad} is tight in approximation factor under the HS hypothesis.
\begin{theorem}\label{biradlb}
Under the HS hypothesis, any algorithm for Bichromatic Radius that achieves a $(5/3-\delta)$-approximation factor for $\delta>0$ in $m$-edge undirected unweighted graphs requires $m^{2-o(1)}$ time.
\end{theorem}

\begin{proof}
Given an instance $U,V\subseteq \{0,1\}^d$ of OV, let $G(U,V)$ be its OV-graph. Create $G'$ which has the same vertex set as $G(U, V)$ except instead of having a vertex for every $v \in V$ it has two copies $v_1 \in V_1$ and $v_2 \in V_2$.

 The edges for $G'$ are: for $u\in U, c\in C$, we add $(u,c)$ as an edge iff $u[c]=1$. For $v\in V, c\in C$, we add $(c, v_1)$ as an edge iff $v[c]=1$. For each $v \in V$ we add an edge $(v_1, v_2)$.
Set $S=U$ and $T=V_1\cup V_2\cup C$. The number of edges in the graph is $O(nd)$.

Suppose that there is no HS solution, then for all $u\in U$ there is some $v\in V$ so that $u\cdot v=0$ and hence $d(u, v_2)\geq 5$. If there is an HS solution $u\in U$, then for all $t\in T$, $d(u,t)\leq 3$.
\end{proof}

\paragraph*{Directed Bichromatic Diameter}

%~2-eps approx for directed needs n^{2-o(1)} time under OV (orig construction is weighted but can be made unweighted) NIKHIL

The following theorem implies that our $m^{3/2}$ $2$-approximation algorithm for directed Bichromatic Diameter from Theorem~\ref{thm:bidirdiam} has a tight approximation factor under SETH.

\begin{theorem}\label{thm:bidirdiamlb}
Under SETH, any algorithm for directed Bichromatic Diameter that achieves a $(2-\delta)$-approximation factor for $\delta>0$ in $m$-edge  graphs requires $m^{2-o(1)}$ time.
\end{theorem}

\begin{proof}

We will show that under SETH, for any positive integer $\ell$, distinguishing between Bichromatic Diameter $\ell+1$ and $2\ell+1$ requires $m^{2-o(1)}$ time. 

Given an instance $U,V\subseteq \{0,1\}^d$ of OV, let $G(U,V)$ be its OV-graph. Create $G'$ which has the same vertex set as $G(U, V)$ except instead of having one vertex for every $v \in V$ it has $\ell$ copies $v_i \in V_i$ for $1 \leq i \leq \ell$. It also has $\ell-2$ additional vertices: $P = \{p_1, p_2,\ldots, p_{\ell-2}\}$.

The edges of $G'$ are: for $u\in U, c\in C$, we add $(u,c)$ as an edge iff $u[c]=1$, and for $c\in C, v\in V$, we add $(c,v_1)$ as an edge iff $v[c]=1$. We add a matching going from $V_i$ to $V_{i+1}$ where edges join the nodes which are copies of each other. For each $c\in C$, we add an edge $(c, p_1)$. We add a path from $p_1$ to $p_{\ell-2}$. For each $u\in U$, we add an edge $(p_{\ell-2}, u)$.
Set $S=U, T=C\cup P\cup V_1\cup V_2 \ldots V_\ell$. The number of edges in the graph is $O(nd)$.

Consider any $u\in U$. By construction, $d(u, z)\leq \ell+1$ for $z \in C \cup P$.
Suppose that there is no OV solution, then for all $u\in U, v\in V$, $u\cdot v \neq 0$ and hence $d(u,v_i)\leq \ell+1$. If there is an OV solution $u\in U, v\in V$, then, $d(u,v_\ell)\geq 2\ell+1$ as the only path is through $P$.
\end{proof}

\paragraph*{Directed Bichromatic Eccentricities}

\begin{proposition}\label{prop:bidirecclb}
Under SETH, any algorithm for Bichromatic Eccentricities that achieves a finite approximation factor in $m$-edge directed graphs requires $m^{2-o(1)}$ time.
\end{proposition}

\begin{proof}
Given an instance $U,V\subseteq \{0,1\}^d$ of OV, let $G(U,V)$ be its OV-graph. Now, direct the edges from $U$ to $C$ and from $C$ to $V$ and set $S=U\cup C$, $T=V$. Notice this is an instance of Bichromatic Eccentricities.

Now, for every $u\in U,v\in V$, if $u\cdot v\neq 0$, $d(u,v)=2$ and if $u\cdot v=0$, $d(u,v)=\infty$ as there is no path from $u$ to $v$. Thus, if there is an OV pair, then the $ST$-Eccentricity for every $u\in U\subseteq S$ is $\infty$, and otherwise it is $2$. Any finite approximation to the $ST$-Eccentricities can distinguish between $\infty$ and $2$, and thus can solve OV. (Notice, we do not even need the Eccentricities of nodes in $C$.) Thus, there can be no $m^{2-\eps}$ time algorithm for $\eps>0$ that achieves a finite approximation factor if SETH holds.
\end{proof}

\paragraph*{Directed Bichromatic Radius}

\begin{proposition}\label{prop:bidirradlb}
Under the HS hypothesis, any algorithm for Bichromatic Radius that achieves a finite approximation factor in $m$-edge directed graphs requires $m^{2-o(1)}$ time. 
\end{proposition}

\begin{proof}
The proof is similar to that for Bichromatic Eccentricities.
Given an instance $U,V\subseteq \{0,1\}^d$ of HS, let $G(U,V)$ be its OV-graph. Now, direct the edges from $U$ to $C$ and from $C$ to $V$, and add an extra node $z$ so that for every $u\in U$ there is a directed edge $(u,z)$. Set $S=U\cup C$, $T=V\cup \{z\}$.

First, if the $ST$-Radius is finite, the $ST$-center (the node achieving the Radius) must be in $U$, since no node in $C$ can reach $z$, by construction. The distance $d(u,z)$ is $1$ for all $u\in U$.
For every $u\in U,v\in V$, if $u\cdot v\neq 0$, $d(u,v)=2$ and if $u\cdot v=0$, $d(u,v)=\infty$ as there is no path from $u$ to $v$. Thus, if there is a HS solution, then the $ST$-Radius is $2$, and otherwise it is $\infty$. Any finite approximation to the $ST$-Radius can distinguish between $\infty$ and $2$, and thus can solve HS. Thus, there can be no $m^{2-\eps}$ time algorithm for $\eps>0$ that achieves a finite approximation factor if the HS hypothesis holds.
\end{proof}

\subsection{$ST$-Diameter, Eccentricities, and Radius}
% We begin with the following graph construction used in prior work.
% Given two sets $U,V\subseteq \{0,1\}^d$, we create their OV graph $G(U,V)$ as follows. The vertices of $G(U,V)$ are $U\cup V\cup [d]$. For $x\in U\cup V$ and $c\in [d]$ there is an (undirected) edge $(x,c)$ if and only if $x[c]=1$. If $|U|=|V|=n$, the number of vertices is $O(n+d)$ and the number of edges is $O(nd)$.
% The main property of $G(U,V)$ is: for $u\in U, v\in V$, if $u\cdot v\neq 0$, $d(u,v)=2$, and otherwise, $d(u,v)\geq 4$.

% From the OV graph we can easily derive several simple results.
% I commented out this above bc it's from the prelims
\paragraph*{Undirected $ST$-Diameter and Eccentricities}

For undirected graphs, Backurs et al.~\cite{diamstoc18} give a time-accuracy trade-off lower bound for $ST$-Diameter that immediately extends to $ST$-Eccentricities (since any Eccentricities algorithm gives a Diameter algorithm with the same running time and accuracy by taking the maximum of Eccentricities).

The following theorem shows that our algorithms for $ST$-Eccentricities from Theorem~\ref{STecc2approx-edge} and Proposition~\ref{STecc3approx} are tight under SETH.

\begin{theorem}[\cite{diamstoc18}]\label{thm:stecclb}
Under SETH, for every $k\geq 2$, every algorithm for $ST$-Diameter (and thus $ST$-Eccentricities) that achieves a $((4k-3)/(2k-1)-\delta)$-approximation for $\delta>0$ in undirected unweighted graphs requires $m^{1+1/(k-1)-o(1)}$ time. 
\end{theorem}

In particular, setting $k=2$ and $3$ in Theorem~\ref{thm:stecclb} shows that our $m^{3/2}$ time $2$-approximation algorithm for $ST$-Eccentricities from Theorem~\ref{STecc2approx-edge} is tight  under SETH, in terms of both approximation factor and runtime. Furthermore, setting $k$ to be arbitrarily large implies that our $\tilde{O}(m)$ time 3-approximation algorithm for $ST$-Eccentricities from Proposition~\ref{STecc3approx} is tight under SETH.

% \begin{proposition}
% Under SETH, any algorithm for $ST$-Eccentricities that achieves a $3-\delta$ approximation factor for $\delta>0$ in $m$-edge undirected graphs requires $m^{1+o(1)}$ time.
% \end{proposition}

% \begin{proposition}
% Under SETH, any algorithm for $ST$-Eccentricities that achieves a $2-\delta$ approximation factor for $\delta>0$ in $m$-edge undirected graphs requires $m^{2-o(1)}$ time.
% \end{proposition}
\paragraph*{Undirected $ST$-Radius}

The following proposition shows that our $\tilde{O}(m^{3/2})$ time 2-approximation algorithm for undirected $ST$-Radius from Theorem~\ref{STradius2approx-edge} has a tight approximation factor under the HS hypothesis.

\begin{proposition}\label{2LBradius}
Under the HS hypothesis, any algorithm for $ST$-Radius that achieves a $(2-\delta)$-approximation for $\delta>0$ in $m$-edge undirected graphs requires $m^{2-o(1)}$ time. 
\end{proposition}

\begin{proof}
Given an instance $U,V\in \{0,1\}^d$ of HS, let $G$ be the OV-graph defined on this instance. Let $S=U$ and $T=V$. Suppose that there is a node $u\in U$ which is not orthogonal to any node in $V$. Then for each $v\in V$, $d(u,v)=2$ by using the coordinate node on which both $u$ and $v$ are $1$. So in this case the $ST$-Radius is $2$. Suppose on the other hand that no such node in $U$ exists, so that for each node $u\in U,$ there is a node $v\in V$ such that $u\cdot v=0$. Then $d(u,v)\ge 4$. Since $S=U$, the $ST$-Radius is at least $4$ in this case. 

So any $(2-\delta)$-approximation algorithm can distinguish between $ST$-Radius $2$ and $4$, and thus solve HS. Therefore, there can be no $m^{2-\epsilon}$ time algorithm for $\epsilon>0$ that achieves a $(2-\delta)$-approximation factor if HS hypothesis holds.
\end{proof}

\paragraph*{Directed $ST$-Diameter}

\begin{proposition}
Under SETH, any algorithm for $ST$-Diameter that achieves a finite approximation factor in $m$-edge directed graphs requires $m^{2-o(1)}$ time. 
\end{proposition}

\begin{proof}
Given an instance $U,V\subseteq \{0,1\}^d$ of OV, let $G(U,V)$ be its OV-graph. Now, direct the edges from $U$ to $C$ and from $C$ to $V$ and set $S=U$, $T=V$.

Now, for every $u\in U,v\in V$, if $u\cdot v\neq 0$, $d(u,v)=2$ and if $u\cdot v=0$, $d(u,v)=\infty$ as there is no path from $u$ to $v$. Thus, if there is an OV pair, then the $ST$-Diameter is $\infty$, and otherwise it is $2$. Any finite approximation to the $ST$-Diameter can distinguish between $\infty$ and $2$, and thus can solve OV. Thus, there can be no $m^{2-\eps}$ time algorithm for $\eps>0$ that achieves a finite approximation factor if SETH holds.
\end{proof}

\paragraph*{Directed $ST$-Eccentricities and Radius}

Propositions~\ref{prop:bidirecclb} and~\ref{prop:bidirradlb} immediately carry over to Directed $ST$-Eccentricities and $ST$-Radius since the Bichromatic version is a special case of the $ST$ version. We state the results here for convenience.

\begin{proposition}\label{prop:stdirecclb}
Under SETH, any algorithm for $ST$-Eccentricities that achieves a finite approximation factor in $m$-edge directed graphs requires $m^{2-o(1)}$ time.
\end{proposition}

\begin{proposition}\label{prop:stdirradlb}
Under the HS hypothesis, any algorithm for $ST$-Radius that achieves a finite approximation factor in $m$-edge directed graphs requires $m^{2-o(1)}$ time. 
\end{proposition}

\subsection{Subset Diameter, Eccentricities, and Radius}
\paragraph*{Subset Diameter and Eccentricities}

The following proposition implies that our $\tilde{O}(m)$ time 2-approximation algorithm for Subset Diameter from Proposition~\ref{prop:subdiam} is tight under SETH, and that our near-linear time almost 2-approximation algorithm for Subset Eccentricities from Theorem~\ref{thm:subecc} is essentially tight under SETH.

\begin{proposition}\label{prop:subdiamlb}
Under SETH, any algorithm for Subset Diameter (and thus Subset Eccentricities) that achieves a $(2-\delta)$-approximation factor for $\delta>0$ in $m$-edge directed graphs requires $m^{2-o(1)}$ time.
\end{proposition}

\begin{proof}
Given an instance $U,V\in \{0,1\}^d$ of OV, we begin with the OV-graph defined on this instance. We add a vertex $u$ adjacent to every vertex in $U$ and a vertex $v$ adjacent to every vertex in $V$. Let $S=U\cup V$. 

If there is no OV solution, every pair of vertices $s\in U$, $s'\in V$ $d(s,s')=2$. Also, every pair of vertices $s,s'\in U$ or $s,s'\in V$ has $d(s,s')=2$ due to the addition of the vertices $u$ and $v$.

On the other hand, if there is an OV solution, in the original OV-graph there exists $s\in U$, $s'\in V$ such that $d(s,s')=4$. We note that the addition of the vertices $u$ and $v$ does not change this fact. 
\end{proof}

% \begin{corollary}
% Under SETH, any algorithm for Subset Eccentricities that achieves a $(2-\delta)$-approximation factor for $\delta>0$ in $m$-edge directed graphs requires $m^{2-o(1)}$ time.
% \end{corollary}

\paragraph*{Subset Radius}

The following proposition implies that our $\tilde{O}(m)$ time 2-approximation algorithm for Subset Radius from Proposition~\ref{prop:subrad} is tight under the HS hypothesis.

\begin{proposition}\label{prop:subradlb}
Under the HS hypothesis, any algorithm for Subset Radius that achieves a $(2-\delta)$-approximation factor for $\delta>0$ in $m$-edge undirected graphs requires $m^{2-o(1)}$ time.
\end{proposition}

\begin{proof}
Given an instance $U,V\in \{0,1\}^d$ of HS, we begin with the OV-graph $U\cup C\cup V$ defined on this instance. Then we add a vertex $u$ adjacent to every vertex in $U$ and a vertex $v$ adjacent to $u$. Let $S=U\cup V\cup \{v\}$.

If there is no HS solution, then in the original OV-graph, for all $s\in U$, there exists some $s'\in V$ such that $d(s,s')\geq 4$. We note that the addition of the vertices $u$ and $v$ does not change this fact. Furthermore, for all vertices $s\in V$, $d(v,s)=4$. Thus, the Subset Radius is at least 4. 

On the other hand, if there is a HS solution, then there exists a vertex $s\in U$ such that for all vertices $s'\in V$, $d(s,s')=2$. Also, $d(s,v)=2$. Thus, the Subset Radius is 2.
\end{proof}

\subsection{Parameterized Bichromatic Diameter, Eccentricities, and Radius} 
In this section we show that modifications of our lower bound constructions show that our algorithms parameterized by the boundary size $|B|$ for Bichromatic Diameter, Eccentricities, and Radius are conditionally tight. Recall that for undirected graphs, $S'$ is the set of vertices in $S$ that have a neighbor in $T$, $T'$ is the set of vertices in $T$ that have a neighbor in $S$, and $B$ is whichever of $S'$ or $T'$ is smaller in size. Since these our parameterized algorithms for undirected graphs have additive error, instead of showing that e.g. distinguishing between values 2 and 3 is hard, we will give results of the form ``for all $\ell$, distinguishing between e.g. $2\ell$ and $3\ell$ is hard". This proves that even algorithms with constant additive error cannot achieve a better multiplicative approximation factor than e.g. $3/2$. 

\paragraph*{Undirected Parameterized Bichromatic Diameter}

The following theorem implies that the multiplicative factor in our $\tilde{O}(m|B|)$ time almost $3/2$-approximation algorithm for undirected Bichromatic Diameter from Theorem~\ref{thm:paramdiam} is tight under SETH for $|B|=\omega(\log n)$.
%in the sense that we have a $\tilde{O}(m|B|)$ time algorithm that gives a multiplicative $3/2$-approximation with constant additive error, and the lower bound construction shows that any algorithm that achieves a better multiplicative approximation factor on graphs with $|B|=\Omega(\log n)$ must run in $\Omega(n^2-o(1))$ time.

\begin{theorem}\label{thm:paramdiamlb}
For any integer $\ell>0$, under SETH any algorithm for Bichromatic Diameter in undirected unweighted graphs that distinguishes between Bichromatic Diameter $4\ell$ and $6\ell$
 requires $m^{2-o(1)}$ time, even for graphs with $|B|=d=\tilde{O}(1)$. 
\end{theorem}

\begin{proof}$ $
\paragraph*{Construction} Given an instance $U,V\in \{0,1\}^d$ of OV, we begin with the OV-graph $U$, $C$, $V$ defined on this instance. We add a new set $U'$ of $n$ vertices, one vertex for each vector in $U$, and connect each vertex in $U$ to its corresponding vertex in $U'$ to form a matching. Symmetrically, we add a new set $V'$ of $n$ vertices, one vertex for each vector in $V$, and connect each vertex in $V$ to its corresponding vertex in $V'$ to form a matching. Then we subdivide each of the edges in the graph into a path of length $\ell$. Let $T$ contain $C\cup V\cup V'$ as well as the vertices on the subdivision paths from $C$ to $V$ and from $V$ to $V'$. Let $S$ be the remaining vertices, that is, $S$ contains $U$, $U'$, the vertices that subdivide the edges between $U$ and $U'$, and the vertices that subdivide the edges between $U$ and $C$.
\paragraph*{Analysis} We note that $T'=C$ and $|C|=d$ so $|B|=d=\tilde{O}(1)$. 

If the OV instance has no solution then for every pair of vertices $u\in U$, $v\in V$, $d(u,v)=2\ell$. Every vertex in $S$ is at most distance $\ell$ from some vertex in $U$ and every vertex in $T$ is at most distance $\ell$ from some vertex in $V$ so the Bichromatic Diameter is at most $4\ell$.

Suppose the OV instance has a solution $u\in U$, $v\in V$. We know that $d(u,v)\geq 4\ell$. Let $u'$ be the vertex in $U'$ that is matched to $u$ and let $v'$ be the vertex in $V'$ that is matched to $v$. We claim that $d(u',v')\geq 6\ell$. Since $U,U'$ and $V,V'$ form matchings the only paths between $u'$ and $v'$ contain $u$ and $v$. Thus, $d(u',v')=d(u',u)+d(u,v)+d(v,v')\geq 6\ell$.
\end{proof}

\paragraph*{Undirected Parameterized Bichromatic Eccentricities}

The following proposition implies that the multiplicative factor in our $\tilde{O}(m|B|)$ time almost $5/3$-approximation algorithm for undirected Bichromatic Eccentricities from Theorem~\ref{thm:paramecc} is tight under SETH for $|B|=\omega(\log n)$.

\begin{proposition}\label{prop:param_ecc}
For any integer $\ell>0$, under SETH any algorithm for Bichromatic Eccentricities in undirected unweighted graphs that distinguishes for all vertices $v$ between $\eps_{ST}(v)=3\ell$ and $\eps_{ST}(v)=5\ell$
 requires $m^{2-o(1)}$ time, even for graphs with $|B|=d=\tilde{O}(1)$. 
\end{proposition}

\begin{proof}$ $
\paragraph*{Construction} Given an instance $U,V\in \{0,1\}^d$ of OV, we begin with the OV-graph $U$, $C$, $V$ defined on this instance. We add a new set $V'$ of $n$ vertices, one vertex for each vector in $V$, and connect each vertex in $V$ to its corresponding vertex in $V'$ to form a matching. Then we subdivide each of the edges in the graph into a path of length $\ell$. Let $S$ contain $U$, $C$, and the vertices that subdivide the edges between $U$ and $C$. Let $T$ contain the remaining vertices. 

\paragraph*{Analysis} We note that $S'=C$ and $|C|=d$ so $|B|=d=\tilde{O}(1)$. 

If there is no OV solution, then for all pairs of vertices $u\in U$, $v\in V$, $d(u,v)=2\ell$. Every vertex in $T$ is of distance at most $\ell$ from some vertex in $T$ so for all vertices $u\in U$, $\eps_{ST}(u)\leq 3\ell$.

If there is an OV solution $u\in U$, $v\in V$, $d(u,v)\geq 4\ell$. Let $v'\in V'$ be the vertex matching to $v$. Then, $d(u,v')\geq 5\ell$ so $\eps_{ST}(v)\geq 5\ell$.
\end{proof}

\paragraph*{Undirected Parameterized Bichromatic Radius}

The following theorem implies that the multiplicative factor in our $\tilde{O}(m|B|)$ time almost $3/2$-approximation algorithm for undirected Bichromatic Radius from Theorem~\ref{thm:paramrad} is tight under the HS hypothesis for $|B|=\omega(\log n)$.

\begin{theorem}\label{thm:paramradlb}
For any integer $\ell>0$, under the HS hypothesis any algorithm for Bichromatic Radius in undirected unweighted graphs that distinguishes between Bichromatic Radius $4\ell$ and $6\ell$
 requires $m^{2-o(1)}$ time, even for graphs with $|B|=d=\tilde{O}(1)$. 
\end{theorem}

\begin{proof}$ $
\paragraph*{Construction} Given an instance $U,V\in \{0,1\}^d$ of HS, we begin with two copies of the construction from Theorem~\ref{thm:paramdiamlb}, $U'_1$, $U_1$, $C_1$, $V_1$, $V'_1$, and $U'_2$, $U_2$, $C_2$, $V_2$, $V'_2$. We then merge each vertex in $U'_1$ with its corresponding vertex in $U'_2$. 

\paragraph*{Analysis} We note that $T'=C_1\cup C_2$ and $|C_1|=|C_2|=d$ so $|B|=2d=\tilde{O}(1)$. 

It will be convenient to imagine that the graph is layered from left to right as $V'_2$, $V_2$, $C_2$, $U_2$, $U'_1$, $U_1$, $C_1$, $V_1$, $V'_1$.

If there is no HS solution, then for all $u_1\in U_1$, there exists some $v_1\in V_1$ such that $d(u_1,v_1)\geq 4\ell$ and for all $u_2\in U_2$, there exists some $v_2\in V_2$ such that $d(u_2,v_2)\geq 4\ell$. Let $u$ be any vertex in $S$ that lies in $U'_1$ or to the right of $U'_1$. Since any path $u$ to a vertex in $V'_2$ contains a vertex in $U_2$, there exists $v\in V'_2$ such that $d(u,v)\geq 6\ell$. Symmetrically, if $u$ is a vertex in $S$ that lies to the left of $U'_1$, there exists $v\in V'_1$ such that $d(u,v)\geq 6\ell$. Thus, the Bichromatic Radius is at least $6\ell$.

On the other hand, if there is a HS solution, then there exists a vertex $u\in U_1$ such that for all vertices $v\in V_1$, $d(u,v)=2\ell$. Let $u'$ be the vertex in $U'_1$ matched to $u$ and let $u''$ be the vertex in $U_2$ matched to $u'$. Then, for all vertices $v\in V_2$, $d(u'',v)=2\ell$. Thus, for all vertices $v\in V_1\cup V_2$, $d(u',v)=3\ell$, so for all vertices $v\in T$, $d(u',v)\leq 4\ell$. Thus, the Bichromatic Radius is at most $4\ell$.
\end{proof}

\paragraph*{Directed Parameterized Bichromatic Diameter}

Recall that for directed graphs, $S'$ is the set of vertices in $S$ with an outgoing edge to a vertex in $T$, $T'$ is the set of vertices in $T$ with an incoming edge from a vertex in $S$, and $B'=S'\cup T'$. We will show that the construction from Theorem~\ref{thm:paramdiamlb} can be made to have small $B'$ (i.e. small $S'$ and $T'$), with a slight additive cost to the Diameter values. The construction will remain undirected.

The following proposition implies that the multiplicative factor in our $\tilde{O}(m|B'|)$ time almost $3/2$-approximation algorithm for Directed Bichromatic Diameter from Theorem~\ref{thm:paramdirdiam} is tight under SETH for $|B'|=\omega(\log n)$.

\begin{proposition}\label{prop:paramdirdiamlb}
For any integer $\ell>0$, under SETH any algorithm for Bichromatic Diameter in directed unweighted graphs that distinguishes between Bichromatic Diameter $4\ell+1$ and $6\ell+1$
 requires $m^{2-o(1)}$ time, even for graphs with $|B|=d=\tilde{O}(1)$. 
\end{proposition}

\begin{proof}$ $
\paragraph*{Construction} We begin with the construction from Theorem~\ref{thm:paramdiamlb}. We replace each vertex $c\in C$ by a pair of vertices $c_1$, $c_2$ and let $(c_1,c_2)$ be an edge. Let $C_1$ and $C_2$ be the set of all $c_1$'s and $c_2$'s respectively. That is, $C_1$ and $C_2$ form a matching. For every edge originally between $u\in U$ and $c\in C$, we replace it with the undirected edge $(u,c_1)$ and for every edge originally between $c\in C$ and $v\in V$, we replace it with the undirected edge $(c_2,v)$. 

\paragraph*{Analysis} The correctness follows from the analysis of Theorem~\ref{thm:paramdiamlb}. Here, we get $4\ell+1$ and $6\ell+1$ instead of $4\ell$ and $6\ell$ due to the addition of the matching between $C_1$ and $C_2$.
\end{proof}

%\paragraph*{Directed Parameterized Bichromatic Eccentricities and Radius}

%Note: I just realized it could be possible to get finite approxes for parameterized bichrom ecc and radius bc the LB construction has large boundary.

%Given a $k$-OV instance consisting of sets $W_0,W_1,\dots,W_{k-1}$, we begin with the $k$-OV-graph $L_0\cup\dots \cup L_k$. Then, we subdivide each edge into a path of length $\ell$. Additionally, we add $k-2$ new layers of vertices $L_{k+1},\dots,L_{2k-2}$, where each new layer contains $n^{k-1}$ vertices and is connected to the previous layer by a matching. That is, each new layer contains one vertex for every tuple $(a_1,\dots,a_{k-1})$ where $a_i\in W_i$ for all $i$, and each $(a_1,\dots,a_{k-1})\in L_j$ is connected to its counterpart $(a_1,\dots,a_{k-1})\in L_{j-1}$ by an edge, for all $j$. We also add a new layer of vertices $L_{-1}$ connected to $L_0$ by a matching. Lastly, we replace each edge in the graph by a path of length $\ell$. Let $S$ contain $L_{-1}$, $L_0$, and all of the vertices introduced by subdivision of edges incident to $L_0$. $T$ contains the remaining vertices in the graph: $L_1,\dots,L_{2k-2}$ and all of the vertices introduced by subdivision of edges between these layers. 

%trade-off LB for ST diameter implies 3-approx is tight for ST eccentricities Mina -- this is just a comment though

\paragraph{Acknowledgements} The authors would like to thank Arturs Backurs for discussions during the early stages of this work.

\bibliographystyle{plain}
\bibliography{references}
\end{document}